%% file: main_arxiv.tex
\documentclass[11pt,letter]{article}
\usepackage{amsmath}
\usepackage{xspace}
\usepackage{todonotes}
\usepackage{amsthm}
\usepackage{fullpage}
\usepackage{hyperref}

\usepackage[utf8]{inputenc}
\usepackage{indentfirst,verbatim} 
\usepackage{listings} 
\usepackage[normalem]{ulem}
\usepackage{soul}
\usepackage{amsfonts}
\usepackage{amsthm}
\usepackage{enumitem}
\usepackage{hyperref}

\usepackage{amsmath,amsfonts,amsthm}
\usepackage{mathbbol}
\usepackage{amsthm}
\usepackage{graphicx,color}
\usepackage{boxedminipage}
\usepackage[linesnumbered, boxed]{algorithm2e}
    \SetKw{Continue}{continue}
    \SetKw{Break}{break}
    \SetKwInput{KwData}{Input}
    \SetKwInput{KwResult}{Output}
\usepackage{framed}
\usepackage{thmtools}
\usepackage{thm-restate}
\usepackage{xspace}
\usepackage{todonotes}
\usepackage{pgfplots}
\usetikzlibrary{calc}
\usetikzlibrary{shapes}
\usetikzlibrary{arrows.meta}
\usepackage{colortbl}
\usepackage{tcolorbox} 

\usepackage{xifthen}
\usepackage{tabularx}
\usepackage{capt-of}
\usepackage{thmtools}
\usepackage{thm-restate}
\usepackage{subcaption}
\usepackage{booktabs}
\usepackage{calc}
\usepackage{cleveref}

\usepackage[margin=1in]{geometry} 

\theoremstyle{plain}

\newtheorem{lemma}{Lemma} 
\newtheorem{claim}{Claim} 
\newtheorem{proposition}{Proposition} 
\theoremstyle{definition}
\newtheorem{definition}{Definition} 
\newtheorem{rrule}{Reduction Rule} 
\newtheorem{trule}{Solution Subroutine}

\newcommand{\cqed}{\renewcommand{\qedsymbol}{$\lrcorner$}\qed}

\newcommand{\polyn}{n^{\Oh(1)}}
\newenvironment{claimproof}{\medskip\noindent \emph{Proof of Claim~\theclaim.}  }{\hfill\cqed\medskip}

\newlength{\RoundedBoxWidth}
\newsavebox{\GrayRoundedBox}
\newenvironment{GrayBox}[1]%
{\setlength{\RoundedBoxWidth}{.93\textwidth}
	\def\boxheading{#1}
	\begin{lrbox}{\GrayRoundedBox}
		\begin{minipage}{\RoundedBoxWidth}}%
		{   \end{minipage}
	\end{lrbox}
	\begin{center}
		\begin{tikzpicture}%
			\node(Text)[draw=black!20,fill=white,rounded corners,%
			inner sep=2ex,text width=\RoundedBoxWidth]%
			{\usebox{\GrayRoundedBox}};
			\coordinate(x) at (current bounding box.north west);
			\node [draw=white,rectangle,inner sep=3pt,anchor=north west,fill=white] 
			at ($(x)+(6pt,.75em)$) {\boxheading};
		\end{tikzpicture}
\end{center}}     

\newenvironment{defproblemx}[2][]{\noindent\ignorespaces%
	\FrameSep=6pt%
	\parindent=0pt%
	\vspace*{-1.5em}
	\ifthenelse{\isempty{#1}}{%
		\begin{GrayBox}{#2}%
		}{%
			\begin{GrayBox}{#2 parameterized by~{#1}}%
			}
			\begin{tabular*}{\textwidth}{@{\hspace{.1em}} >{\itshape} p{1.8cm} p{0.8\textwidth} @{}}%
			}{
			\end{tabular*}%
		\end{GrayBox}%
		\ignorespacesafterend
	}

	\newcommand{\Oh}{\mathcal{O}}

	\DeclareMathOperator{\Ima}{Im}

	\DeclareMathOperator{\operatorClassNP}{{\sf NP}}
	\newcommand{\classNP}{\ensuremath{\operatorClassNP}}

	\DeclareMathOperator{\operatorClassFPT}{{\sf FPT}\xspace}
	\newcommand{\classFPT}{\ensuremath{\operatorClassFPT}\xspace}
	\DeclareMathOperator{\operatorClassW}{{\sf W}}
	\newcommand{\classW}[1]{\ensuremath{\operatorClassW[#1]}}
	\DeclareMathOperator{\operatorClassParaNP}{{\sf Para-NP}\xspace}
	\newcommand{\classParaNP}{\ensuremath{\operatorClassParaNP}\xspace}
	\DeclareMathOperator{\operatorClassXP}{{\sf X}P\xspace}
	\newcommand{\classXP}{\ensuremath{\operatorClassXP}\xspace}
	\lstnewenvironment{code}{\lstset{language=Haskell,basicstyle=\small}}{}

    \title{\textsc{Path Cover, Hamiltonicity,  and Independence Number: An FPT Perspective}
    \thanks{A preliminary version of this paper appeared in the proceedings of STOC 2026.}
    }

	\author{
		Fedor V. Fomin\thanks{
			Department of Informatics, University of Bergen, Norway.}\\fedor.fomin@uib.no
		\and
		Petr A. Golovach\addtocounter{footnote}{-1}\footnotemark{}\\petr.golovach@uib.no
		\and
		Nikola Jedli\v{c}kov\'{a}\thanks{
		        Department of Applied Mathematics and Department of Algebra,
                Faculty of Mathematics and Physics, Charles University, Czech Republic.}\\nikola.jedlickova@matfyz.cuni.cz
		\and Jan Kratochvíl\thanks{
		        Department of Applied Mathematics, Faculty of Mathematics and Physics, Charles University,   Czech Republic.}\\honza@kam.mff.cuni.cz		 
        \and		
		Danil Sagunov\thanks{
			St.\ Petersburg State University, Russia and
            St.\ Petersburg Department of V.~A.~Steklov Mathematical Institute of the Russian Academy of Sciences, Russia.
        }\\danilka.pro@gmail.com
		\and 
		Kirill Simonov\addtocounter{footnote}{-4}\footnotemark{}\\kirill.simonov@uib.no
	}

	\date{}
	
	\usetikzlibrary{arrows}
	
\begin{document}
		
\maketitle	
\thispagestyle{empty}

\begin{abstract}
 The classic theorem of Gallai and Milgram (1960) generalizes several fundamental results in Graph Theory, such as Dilworth’s theorem on posets and Kőnig's theorem on matchings in bipartite graphs. The theorem asserts that for every graph \(G\),
the vertex set of \(G\) can be partitioned into at most \(\alpha(G)\) vertex-disjoint paths,
where \(\alpha(G)\) is the maximum size of an independent set in \(G\).
The proof of the Gallai-Milgram theorem is constructive and yields a polynomial-time algorithm 
that computes a covering of \(G\) by at most \(\alpha(G)\) vertex-disjoint paths. 
While the Gallai-Milgram theorem is tight---there are graphs where one really needs 
\(\alpha(G)\) paths, not fewer, to cover the vertex set of \(G\)---it was not known prior to our work 
whether deciding if a graph \(G\) could be covered by fewer than \(\alpha(G)\) vertex-disjoint paths 
can be done in polynomial time.

We resolve this question by proving the following algorithmic extension of the Gallai-Milgram theorem for undirected graphs: There is an algorithm that, for an \(n\)-vertex graph \(G\) and an integer parameter \(k \ge 1\), runs in time
$
  2^{2^{\Oh(k^4\log{k})}} \cdot n^{\Oh(1)}
$
and outputs a path cover \(\mathcal{P}\) of \(G\) together with  
\begin{itemize}
    \item   a correct conclusion that \(\mathcal{P}\) is a minimum-size path cover, or
    \item  an independent set of size 
    \(\lvert \mathcal{P}\rvert + k\), certifying that \(\mathcal{P}\) contains 
    at most \(\alpha(G) - k\) paths.
\end{itemize}

Thus, for \(k \in \Oh((\log\log n)^{\frac{1}{4}-\varepsilon})\) our algorithm runs in polynomial time and either computes a minimum-size path cover of \(G\), or finds a path cover of size at most \(\alpha(G) - k\). We find the existence of such an algorithm quite surprising for the following reason. The problems of computing a minimum-size path cover and a maximum independent set are both notoriously hard, yet our algorithm either solves one of them or provides meaningful information about the other.

The proof of our algorithmic extension of the Gallai-Milgram theorem is non-trivial and builds on several novel algorithmic ideas. 
One of the key subroutines in our algorithm is an \(\mathsf{FPT}\) algorithm, parameterized 
by \(\alpha(G)\), for deciding whether \(G\) contains a Hamiltonian path. 
This result is of independent interest---prior to our work
 no polynomial-time algorithm for deciding Hamiltonicity was known even for graphs with independence number at most three. 
Moreover, the algorithmic techniques we develop apply to a wide array of problems in undirected graphs, 
including \textsc{Hamiltonian Cycle}, \textsc{Path Cover}, \textsc{Largest Linkage}, 
and \textsc{Topological Minor Containment}. We show that all these problems 
are \(\mathsf{FPT}\) when parameterized by the independence number of the graph.

Notably, the independence-number parameterization departs from the typical direction of research in parameterized complexity. 
First, \(\alpha(G)\) measures a graph's density, whereas most prior work in the area 
focuses on parameters describing sparsity, such as treewidth or vertex cover. 
Second, most structural parameters studied in parameterized complexity can be computed exactly or well-approximated in polynomial or even \(\mathsf{FPT}\) time, whereas computing \(\alpha(G)\) is notoriously difficult from almost any computational perspective. 
The fact that it can nevertheless serve as the basis for efficient parameterization is particularly striking.
\end{abstract}

\tableofcontents
\newpage
\section{Introduction}
The Gallai–Milgram theorem, published by Tibor Gallai and Arthur Milgram in 1960 \cite{GallaiM60}, is a cornerstone of structural graph theory and can be viewed as a far-reaching generalization of several classical results, such as Dilworth's and Kőnig's theorems \cite{Dilworth1950,konig1931}, concerning the covering of graphs. In its original form, the theorem states that every (directed or undirected) graph \(G\) can be covered by at most \(\alpha(G)\) vertex-disjoint paths, where \(\alpha(G)\) is the size of a maximum independent set in \(G\).

Beyond its immediate implications, the Gallai–Milgram theorem belongs to a broader line of research in graph theory focused on connections between Hamiltonicity and independence. This theme can be traced back to the seminal work of Rédei \cite{Redei1934}, who showed that every tournament (a directed graph with \(\alpha(G) = 1\)) contains a directed Hamiltonian path. Further investigations into the relationships among stability number, connectivity, and Hamiltonicity began with Nash-Williams \cite{NashWill71} and Chvátal–Erdős \cite{chvatal1972note}, and have continued in later works (see, e.g., \cite{AmarFournierGerma1983,bauer1989generalization,jackson1987chvatal,Koider1994,BenArroyoHartman1988,BessyThomasse2007,HahnJackson1990,thomasse2001covering,fox2009paths}).

In this paper, we present a novel algorithmic perspective on this classic theme in graph theory: the interplay among the Gallai–Milgram theorem, Hamiltonicity, and the independence number.

\subsection{Our contribution} 
Known constructive proofs of the Gallai-Milgram theorem can be used to construct in polynomial time a covering of the vertex set of a graph $G$ by at most $\alpha(G)$ vertex-disjoint paths. Our following theorem presents a significant generalization of this algorithmic result for undirected graphs.

\begin{restatable}{theorem}{thmabovgm}
    \label{thm:above_GM}
    There is an algorithm that, given an $n$-vertex graph $G$  and an integer parameter $k \ge 1$, in time $2^{k^{\Oh(k^4)}} \cdot n^{\Oh(1)}$ outputs a cover of $V(G)$ by a family $\mathcal{P}$  of vertex-disjoint paths and, furthermore,
    \begin{itemize}
    \item either correctly reports that  $\mathcal{P}$ is a minimum-size path cover, or
 \item  outputs together with $\mathcal{P}$ an independent set of size at least $|\mathcal{P}|+k$ certifying that  $\mathcal{P}$ contains at most $\alpha(G)-k$ paths.	   
    \end{itemize}
    \end{restatable}    
    
  Loosely speaking, the theorem means that computing a path cover below the Gallai--Milgram bound, that is, a path cover of size at most \(\alpha(G)-k\), is \classFPT when parameterized by \(k\). However, formally, such a statement is not correct: A trivial reduction shows that deciding whether \(G\) has a path cover of size at most \(\alpha(G)-k\) is \classW{1}-hard. Indeed, for an \(n\)-vertex input graph \(G\), construct a graph \(G'\) by adding a clique of size \(n\) whose vertices are made universal to \(G\). Then \(\alpha(G') = \alpha(G)\), and \(G'\) is trivially Hamiltonian. Thus, for \(k' = k-1\), the graph \(G'\) can be covered by at most \(\alpha(G')-k'\) paths if and only if \(\alpha(G) \geq k\). Since the \textsc{Independent Set} problem is \classW{1}-complete~\cite{DowneyF99} when parameterized by the solution size, the existence of an \classFPT algorithm in \(k\) for deciding whether \(G\) has a path cover of size at most \(\alpha(G) - k\) would contradict the basic complexity assumptions of parameterized complexity theory. 

Of course, in situations where the independence number of a graph can be efficiently computed, for example when the input graph is perfect, \Cref{thm:above_GM} implies that the path cover below Gallai--Milgram is \classFPT. The interesting feature of \Cref{thm:above_GM} is that it circumvents the intractability of computing \(\alpha(G)\)---the algorithm does not need to know \(\alpha(G)\) at all. This is achieved by allowing the computation of an optimal path cover rather than concluding that \(G\) has no path cover of size at most \(\alpha(G)-k\). Another notable feature of the algorithm is that it \emph{always} outputs a path cover which is either optimal or of size at most \(\alpha(G)-k\). It is somewhat ironic that the ``formal'' intractability of the path cover below Gallai--Milgram arises only from the requirement of identifying no-instances, that is, situations where \(\alpha(G)-k\) is smaller than the size of a minimum path cover.

\Cref{thm:above_GM} aligns with a broader trend in parameterized algorithms that seeks algorithmic analogues of classic theorems in extremal graph theory. Examples include algorithmic extensions of the Poljak--Turzík theorem~\cite{MnichPSS14}, Dirac's theorem~\cite{FominGSS21,DBLP:conf/wg/Jansen0N19}, and the Hajnal-Szemerédi theorem~\cite{gan2023algorithmic}. As with \Cref{thm:above_GM}, each of these results provides an \classFPT algorithm parameterized “above” or “beyond” a known combinatorial bound. A distinguishing feature of \Cref{thm:above_GM}, however, is that in previous work the relevant combinatorial bounds (e.g., the maximum vertex degree) were efficiently computable, whereas determining whether \(\alpha(G) \ge k\) is \classW1-hard with respect to \(k\)~\cite{DowneyF99} and likely does not admit an \classFPT approximation~\cite{LinRSW22,SK22}.  

\medskip

One of the key components in the proof of \Cref{thm:above_GM} is an \classFPT\ algorithm, parameterized by \(\alpha(G)\), for deciding the Hamiltonicity of \(G\).  
Moreover, the techniques developed to obtain this result extend well beyond Hamiltonicity and apply to a broader class of problems.  
Our second main result, \Cref{thm:max_model}, provides an \classFPT\ algorithm for computing various topological embeddings in a graph, parameterized by its independence number. While only a restricted case of this result is needed in the proof of \Cref{thm:above_GM}, we believe that the general approach is of independent interest, and therefore present it in full generality.

\medskip

Let us recall that a graph $H$ is called a \emph{topological minor} of a graph $G$ if there exists a subgraph of $G$ that is isomorphic to a subdivision of $H$.  
A \emph{subdivision} of a graph is formed by replacing some (maybe zero) edges with simple paths (i.e., inserting additional vertices along those edges).  
Thus, saying that $H$ is a topological minor of $G$ means that $G$ contains a subgraph $M$ that represents a ``subdivided'' version of $H$, which we refer to as a \emph{topological minor model} $M$ of $H$ in $G$.  
Many interesting problems can be expressed in terms of topological minors.  
For example, the problem of computing a longest cycle in a graph is equivalent to computing the largest model $M$ of a graph $H = C_3$, the cycle on three vertices. More formally,

 \begin{definition}[TM-model and TM-embedding]
Let $G$, $H$ be graphs, let $M$ be a subgraph of $G$ and $T \subseteq V(M)$ be a vertex set such that all vertices of $M- T$ have degree two in $M$. We say that $(M, T)$ is a \emph{topological minor model} (TM-model) of $H$ in $G$ if $H$ is isomorphic to the
graph obtained from $M$ by dissolving in $M$ all the vertices of $M- T$. Here, \emph{dissolving} a vertex $v$ of degree two with distinct neighbors means deleting it while making its neighbors adjacent.
Note that the operation is unfeasible if the neighbors of $v$ are adjacent, and it should not be the case for the vertices of $M-T$.
We further say that for the injective mapping $f: V(H) \to V(G)$ induced by the isomorphism above, $(M, f)$ is a \emph{topological minor embedding} (TM-embedding) of $H$ in $G$. The \emph{size} of a TM-model $(M, T)$ and of a TM-embedding $(M, f)$ is the number of vertices in $M$.
\end{definition}
 
 In other words, in a TM-embedding, 
 the vertices of $H$ are mapped to distinct vertices of $G$ such that the edges of $H$ could be mapped to internally vertex-disjoint paths of $G$ connecting the corresponding pairs of vertices of $G$. Moreover, the mapping from edges of $H$ to paths in $G$ can be ``read off'' from the corresponding model subgraph $M$.

To capture topological minor models containing specific terminal vertices, like  finding a longest cycle containing a specified terminal vertex set $T$, we also consider a more general variant of TM-embedding. In this variant, every vertex of $H$ is assigned a list of vertices of $G$ where it could be mapped. 
 
 \begin{definition}[List TM-embedding] For graphs $H$, $G$, and a mapping $L : V(H) \to 2^{V(G)}$, we say that a TM-embedding $(M,f)$, where $f: V(H) \to V(G)$, is a  \emph{list topological minor embedding} (list TM-embedding), if  for each $h \in V(H)$, $f(h) \in L(h)$.
 \end{definition}
 
The following theorem establishes that computing the list TM-embedding of maximum size is \classFPT, that is, solving the \textsc{Maximum List TM-Embedding} problem,
 parameterized by the size of the maximum independent set in $G$ and the size of the graph $H$. 
The algorithm, for a given parameter $k$,  either computes a list TM-embedding of maximum size in \classFPT time with respect to $k$ and $|H|$, or correctly reports that the input graph does not fall within the appropriate domain, meaning it has an independence number at least $k$. Such algorithms are what Raghavan and Spinrad refer 
to as \emph{robust} \cite{RaghavanS03}.
More precisely, 
 
\begin{restatable}{theorem}{thmmaxmodel}
    \label{thm:max_model}
    There is an algorithm that, given graphs $H, G$,   a list assignment $L : V(H) \to 2^{V(G)}$ and an integer parameter $k \ge 1$, in time $2^{|H|^{\Oh(k)} \cdot k^{\Oh(k^2)}} \cdot|G|^{\Oh(1)}$
    \begin{itemize}
    \item either correctly reports that there is no list TM-embedding of $H$ in $G$, or
    \item computes a list TM-embedding of $H$ in $G$ of the maximum size, where the maximum is taken over all list TM-embeddings of $H$ in $G$, or
    \item outputs an independent set in $G$ of size  $k$.
    \end{itemize}
    \end{restatable}

 Similar to \Cref{thm:above_GM}, the question concerning parameterization by \(\alpha(G)\) deviates from the standard setup used in structural parameterized complexity, where most algorithmic results rely on the fact that the problem of computing or approximating the corresponding graph structural parameter is in \classFPT. For example, deciding whether a graph contains a Hamiltonian cycle is \classFPT when parameterized by the treewidth of the graph, for the following reason. Given an input integer \(k\), there exist algorithms (either exact or approximate) that either correctly decide that the treewidth of the input graph is greater than \(k\), or construct a tree decomposition of width \(\Oh(k)\) in time \(f(k)\cdot n^{\Oh(1)}\)~\cite{Bodlaender96,Korhonen21,RobertsonS86}. The resulting tree decomposition can then be used to decide whether the graph is Hamiltonian in time \(g(k)\cdot n^{\Oh(1)}\). In contrast, in our setting the parameter under consideration does not admit an \classFPT approximation~\cite{LinRSW22,SK22}, which necessitates a different approach.

  \Cref{thm:max_model} implies that \textsc{Maximum List TM-Embedding} is \classFPT\ with respect to the combined parameter \(|H| + \alpha(G)\).  
For any integer \(k > \alpha(G)\), the algorithm of \Cref{thm:max_model} either outputs a maximum-size list TM-embedding or correctly determines that no such embedding exists.  
The fact that \(\alpha(G)\) cannot be computed efficiently can be circumvented by the following trick.  
To find a list TM-embedding of \(H\) in \(G\), we run the algorithm of \Cref{thm:max_model} for increasing values of the parameter \(k\), starting with \(k = 1\), until the algorithm either finds an embedding or concludes that no appropriate embedding exists.

To our knowledge, neither list TM-embeddings nor the \textsc{Maximum List TM-Embedding} problem have been studied so far. However, these problems are natural generalizations of several vital problems in graph algorithms. In particular,  \Cref{thm:max_model} implies the following corollaries.

\begin{itemize}
\item When graph $H$ consists of a path $P_2$ on two vertices whose lists are all vertices of $V(G)$, then this is the problem of finding a longest path in a non-empty graph $G$. If this path contains $|V(G)|$ vertices, this is a Hamiltonian path. For $H=C_3$, we have the problem of finding a  longest cycle in $G$, the generalization of computing a Hamiltonian cycle in $G$. By \Cref{thm:max_model}, these tasks could be performed in time  $f(\alpha(G)) \cdot |G|^{\Oh(1)}$ which implies that the problem of finding a longest path/cycle is \classFPT parameterized by $\alpha(G)$. 

\item  A graph \(G\) is \emph{Hamilton-connected} (also known as Hamiltonian connected) if every two distinct vertices of \(G\) are connected by a Hamiltonian path; see, e.g., \cite[p.~61]{BondyM}.  
Checking whether a graph is Hamilton-connected thus reduces to computing, for each pair \(u, v \in V(G)\) of distinct vertices, a maximum-size TM-embedding of \(H = P_2\) with lists \(\{u\}\) and \(\{v\}\), which can be done in time \(f(\alpha(G)) \cdot |G|^{\Oh(1)}\) by \Cref{thm:max_model}.

\item  In  \textsc{Path Cover}, the task is to cover  all vertices of $G$ by at most $p$ vertex-disjoint paths.  By the Gallai-Milgram theorem, 
 if $p\geq \alpha(G)$ then a covering by at most $p$ paths always exists and we can assume that $p<\alpha(G)$. We can solve the problem by taking $H$ to be disjoint unions of at most $p$ disjoint copies of $K_2$ and $K_1$
 assigning to each vertex of $H$ the list containing all vertices of $V(G)$.  Thus, \Cref{thm:max_model} yields  an algorithm of running time $f(\alpha(G)) \cdot |G|^{\Oh(1)}$ solving  this problem. Of course, it also follows from \Cref{thm:above_GM}.  For \textsc{Cycle Cover}, that is when the graph $H$ consists of at most $p$ disjoint copies of $C_3$, 
we have the running time  $f(\alpha(G)+p) \cdot |G|^{\Oh(1)}$.

\item Consider  the following generalization of  \textsc{Path Cover}. In the \textsc{Hamiltonian-$\ell$-Linkage} problem, we are given $\ell$ pairs of vertices $(s_1,t_1),\dots, (s_\ell, t_\ell) $. The task is to connect these pairs by disjoint paths that altogether traverse all vertices of the graph. For $\ell=1$ this is the problem of deciding whether a graph for selected vertices $s$ and $t$ contains a Hamiltonian  $s$-$t$-path. \textsc{Hamiltonian-$\ell$-Linkage} could be encoded as \textsc{Maximum List TM-embedding} with 
$H$ consisting of   $\ell$ disjoint copies of $P_2$.  Let the vertices of these $P_2$'s be $\{x_i, y_i\}$, $i\in \{1,\dots, \ell\}$. Then the list mapping $L$ assigns $\{s_i\}$ to $x_i$ and $\{t_i\}$ to $y_i$.  By \Cref{thm:max_model}, \textsc{Hamiltonian-$\ell$-Linkage} is solvable in time  $f(\alpha(G)+\ell) \cdot |G|^{\Oh(1)}$.   By playing with different graphs $H$ in \textsc{Maximum List TM-Embedding} and \Cref{thm:max_model}, one easily obtains algorithms computing spanning subgraphs with various properties like a spanning tree with at most $\ell $ leaves. 

\item  Another interesting scenario is the maximization variant of the \textsc{$T$-Cycle} problem. 
In the \textsc{$T$-Cycle} problem, we are given a graph $G$ and a set $T\subseteq V(G)$ of terminals. The task is to decide whether there is a cycle passing through all terminals \cite{BjorklundHT12,Kawarabayashi08,Wahlstrom13} (in \cite{BjorklundHT12}, the authors look for the minimum-length cycle). By guessing additional fictive terminals in the case $|T|<3$, we can assume that $|T|\geq 3$.
Then we can reduce the maximization variant  to finding a list TM-embedding  of maximum size by enumerating all $|T|!$ orderings in which terminal vertices occur in an optimal solution. For each guessed ordering, we apply \Cref{thm:max_model} for $H$ being a cycle on $|T|$ vertices such that each vertex of $H$ is assigned to precisely one terminal vertex, and the cyclic ordering of $H$ corresponds to the guess.  This brings us to the  $   f(\alpha(G)+|T|) \cdot |G|^{\Oh(1)}$ algorithm computing a cycle of maximum length containing all terminals in $T$. 
Similar arguments hold for a more general colored variant of the problem. In this variant, we are given a graph $G$, and some of the vertices of $G$ are assigned colors from $1$ to $L$. Then, the task is to find a longest cycle containing the maximum number of colors \cite{FominGKSS23}.  Note that an optimal cycle can contain many vertices of the same color.  By guessing the colors participating in a  solution cycle $C$ and the ordering of their first appearance in some cyclic ordering of $C$, we construct a cycle $H$ on at most $L$ vertices. Every vertex of $H$ is assigned the list of vertices of $G$ corresponding to the guessed color.  This brings us to the algorithm of running time   $  f(\alpha(G)+L) \cdot |G|^{\Oh(1)}$. 
 \end{itemize}

We remark that the result of \Cref{thm:max_model} is tight in the sense that  \textsc{Maximum List TM-Embedding} is trivially \classParaNP-hard when parameterized by either $\alpha(G)$ or $|H|$ only. The problem is \classNP-hard for $H=K_2$ because of the \classNP-completeness of  \textsc{Hamiltonian Path}~\cite{GareyJ79}. Furthermore, observe that finding a path $P$ with $n$ vertices in an $n$-vertex graph $G$ is equivalent to finding a subgraph isomorphic to the complement of $G$ in the complement of $P$.  Since  \textsc{Hamiltonian Path} is \classNP-complete~\cite{GareyJ79} and the independence number of the complement of $P$ is at most two,  \textsc{Maximum List TM-Embedding} is \classNP-hard for graphs $G$ with $\alpha(G)\leq 2$ when $H$ is a part of the input.

\subsection{Overview of \Cref{thm:above_GM,thm:max_model}}
The proof of  \Cref{thm:above_GM} uses \Cref{thm:max_model} as a black box. We start by sketching the proof of  \Cref{thm:above_GM}.

\medskip\noindent\emph{Overview of the proof of \Cref{thm:above_GM}:}
To streamline the sketch, some steps are simplified; the actual proof differs in some details.  
We first recall the standard argument for the Gallai–Milgram theorem in undirected graphs. The key observation is that if a path cover of \(G\) has more than \(\alpha(G)\) paths, then there must be an edge connecting the endpoints of two distinct paths. By merging these paths, the cover size decreases by one. Thus, starting from the trivial cover consisting of \(\lvert V(G) \rvert\) single-vertex paths, we can iteratively obtain a path cover of size at most \(\alpha(G)\). This procedure runs in polynomial time. Importantly, we do not need to know \(\alpha(G)\) in advance; the algorithm simply keeps merging paths with adjacent endpoints until no such pair remains.

The merging procedure used to prove the Gallai–Milgram theorem does not provide any additional information beyond \(p \le \alpha(G)\), where \(p\) is the size of the constructed path cover \(\mathcal{P}\). To move beyond this bound, we extract more structural information from a path cover once the Gallai–Milgram merging process can no longer continue. For this purpose, we introduce the notions of \emph{usual} and \emph{special} paths. A path is called \emph{usual} if its endpoints are distinct and no edge connects them in \(G\); otherwise, the path is \emph{special}.  

Intuitively, usual paths are “open-ended”: their endpoints are free, and since no edge ties them together, both endpoints can safely be included in an independent set. Usual paths are therefore useful for building independence, as each such path potentially contributes two vertices to the independent set. 
More formally, assume there are at least \(k\) usual paths in \(\mathcal{P}\) and that no edges connect the endpoints of two distinct paths \(P_i, P_j \in \mathcal{P}\). Then, if for each special path (of which there are at most \(|\mathcal{P}|-k = p-k\)) we take one endpoint, and for each usual path we take both endpoints, we obtain an independent set \(I\) of size at least \((p-k)+2k \ge p+k\). The path cover \(\mathcal{P}\) together with the independent set \(I\) forms a certificate in \Cref{thm:above_GM} for the case when \(\mathcal{P}\) contains at most \(\alpha(G)-k\) paths.

Thus, the first step of the proof of \Cref{thm:above_GM} is as follows. We start from a trivial path cover \(\mathcal{P}\) consisting of single-vertex paths and perform “Gallai–Milgram mergings”. That is, we merge paths in \(\mathcal{P}\) exhaustively as long as there exists a pair of paths with adjacent endpoints. When no such merge is possible, if at least \(k\) paths are usual, the algorithm constructs \(I\) and outputs a valid solution.

In the remainder of the proof, the algorithm considers the case when in the constructed path cover \(\mathcal{P}\):
\begin{itemize}
\item[\emph{(i)}]
 \emph{No two endpoints of distinct paths are adjacent}, and
 \item[\emph{(ii)}]
  \emph{fewer than \(k\) paths are usual, so that the remaining paths are special. }
  \end{itemize}

  Informally, special paths are “closed”: they either consist of a single vertex (a trivial path), or have two vertices,
  or are of length at least two and their endpoints are connected by an edge, so that the path together with this edge forms a cycle. In this situation, an edge between any two vertices of different special paths plays the same role as an edge between their endpoints: it allows us to merge the two paths and cover all their vertices with a single path. Therefore, while special paths are not very useful for constructing independent sets, they are extremely convenient for merging. Our algorithm exploits this property by exhaustively merging special paths until no edges remain between them. 

As a result, we obtain the following structural property of the path cover:  
\begin{itemize}
\item[\emph{(iii)}]
\emph{if two distinct paths \(P_i, P_j \in \mathcal{P}\) are special, then there is no edge between \(V(P_i)\) and \(V(P_j)\) in \(G\).}  
  \end{itemize}

We perform a similar merging trick in the case when \(P_i\) is special, \(P_j\) is usual, and there is an edge from \(V(P_i)\) to an endpoint of \(P_j\). This yields the fourth property:  
\begin{itemize}
\item[\emph{(iv)}] \emph{If a path \(P_i \in \mathcal{P}\) is special, then there is no edge between \(V(P_i)\) and any endpoint of any other path \(P_j \in \mathcal{P}\).}  
\end{itemize}

Together, properties (i)–(iv) provide additional structural insight into the path cover \(\mathcal{P}\) and, consequently, into the graph \(G\).
Armed with these properties, our algorithm is now able to construct a larger independent set in $G$ than before.
For each usual path, the algorithm still takes both its endpoints into $I$.
For each special path $P_i\in \mathcal{P}$, the algorithm can now take two arbitrary non-adjacent vertices from $V(P_i)$ into $I$ provided that such vertices exist.
If $V(P_i)$ is instead a clique in $G$, the algorithm takes just a single arbitrary vertex into $I$.
The resulting set $I$ is independent by properties (i), (iii), (iv) of $\mathcal{P}$.
If at least $k$ special paths do not induce a clique in $G$, the size of $I$ is at least $p+k$, and this immediately resolves the problem.
Thus, the only remaining scenario is that, 
\begin{itemize}
\item[\emph{(v)}] 
\emph{except for at most \(k\) special paths, the vertex set of each of the special paths is a clique in \(G\). }
\end{itemize}

 \begin{figure}
	\centering
	(a)\;
	\includegraphics{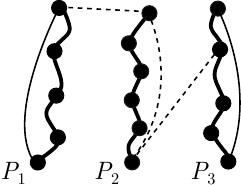}
	\;\;\;\;\;(b)\;
	\includegraphics{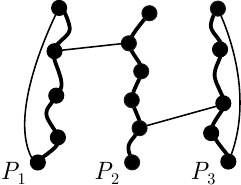}
	\;\;\;\;\;(c)\;
	\includegraphics{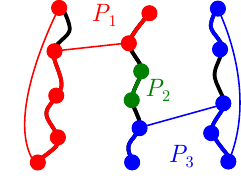}
	\caption{Example of a three-path transformation. In part (a), $P_2$ is a usual path, while $P_1$ and $P_3$ are special paths. 
		Dashed lines represent some of the absent edges.
		Part (b) represents the same paths connected by edges required for the transformation.
		Part (c) shows how new paths are formed. $P_1$ and $P_3$ become usual paths, while $P_2$ can become special or even disappear completely.}\label{fig:intro-three-path-transform}
\end{figure}

However, this structural insight alone addresses only a part of the problem. To progress further, we need a more sophisticated transformation of the paths in \(\mathcal{P}\). Rather than merging two paths, we now consider \emph{two special paths} and  \emph{one usual path} such that there are two edges (with no shared endpoints) connecting the special paths to the usual path. Using these edges, the transformation creates at most three new paths, at least two of which become usual. As a result, the number of special paths in \(\mathcal{P}\) decreases without increasing the total number of paths (see \Cref{fig:intro-three-path-transform}).

If   all these  transformations cannot be applied, we arrive at the following situation. There must exist a vertex set \(S\) of size smaller than \(k\) that separates all usual paths from at least \(p - 4k\) special paths, each of which induces a clique in \(G\). In the auxiliary \Cref{lemma:selected_cliques}, we demonstrate how this small separator \(S\) can be exploited to remove irrelevant cliques from \(G\) without
affecting the difference between the independence number and the size of an optimal path cover.
Applying this polynomial-time procedure ensures that, in the end, at most \(\Oh(k^2)\) special paths remain in \(\mathcal{P}\).

By applying transformations and irrelevant clique removals, we end up in the situation where \(p = \Oh(k^2)\). In this case, we invoke \Cref{thm:max_model} to compute an optimal path cover of \(G\), using \(p + k\) as its parameter. If \Cref{thm:max_model} does not return an optimal path cover, it instead provides an independent set of size \(p + k\). In either outcome, the algorithm obtains a valid solution for \Cref{thm:above_GM}.

All steps of the algorithm are performed in polynomial time, except for the call to \Cref{thm:max_model}, resulting in running time $2^{k^{\Oh(k^4)}} \cdot n^{\Oh(1)}$.

\medskip\noindent\emph{Overview of the proof of \Cref{thm:max_model}:}
We focus on outlining the main ideas of the proof specifically for Hamiltonian-\(\ell\)-linkages. The more general extension to arbitrary list TM-embeddings is similar but requires additional technical details.

To begin, we sketch an \classXP algorithm that solves \textsc{Hamiltonian-\(\ell\)-Linkage} in time \(n^{f(\alpha(G)+\ell)}\). Although this algorithm does not achieve the desired complexity, it is instructive for understanding the proof.
Suppose we are given a parameter \(k \ge \alpha(G)\). (This assumption is acceptable for an \classXP algorithm because \(k\) can be computed by brute force in \(n^{\Oh(\alpha(G))}\) time.) 
The \classXP algorithm is recursive and is inspired by the 
  Chv\'atal-Erd\H{o}s theorem on highly-connected graphs, see \Cref{prop:chvatal}. Informally, the Chv\'atal-Erd\H{o}s theorem states that if a graph lacks a large independent set but is highly-connected, it necessarily contains a Hamiltonian path.
   If its connectivity is low, we find a small vertex cut, delete it, and then recursively process the connected components of the resulting graph. Because a Hamiltonian path might traverse each component multiple times, we find it more convenient to work with {\em Hamiltonian-$\ell$-linkages}, that is,  linkages consisting of vertex-disjoint paths connecting specified pairs of vertices and covering all vertices of the graph. 
To handle such cases, one can strengthen the Chv\'atal-Erd\H{o}s theorem to apply to Hamiltonian-linkages. 
 Following the foundational work of Bollobás and Thomason \cite{BollobasThomason1996}, there have been numerous advances in understanding how linkages relate to connectivity. However, our setting additionally requires that the linkages are Hamiltonian. While no existing results directly address this requirement, we adapt the result of Thomas and Wollan~\cite{Thomas2005} on linkages in highly-connected graphs to show that if a graph \(G\) is \(\Omega(\alpha(G) \cdot \ell)\)-vertex-connected, then it is Hamiltonian-\(\ell\)-linked. In other words, there exists a path cover of size \(\ell\) that links any prescribed set of endpoints.
  
The algorithm proceeds as follows. If the input graph $G$ is $\Omega(k\cdot \ell)$-connected, then  the algorithm correctly reports that $G$ is Hamiltonian-$\ell$-linked.
Otherwise, if the connectivity is lower, there exists a minimum vertex separator  $X$ of $G$ with $\Oh(k\cdot\ell)$ vertices. Each connected component of $G-X$ has an independence number at most $\alpha(G)-1$, and the algorithm proceeds recursively on each component. For that, the algorithm performs an exhaustive search of how the solution paths are split into parts by $X$.
For every connected component of $G-X$ and each split, the algorithm also performs an exhaustive search of entry and exit points of the parts of the paths in the component. 
Since the number of parts is at most $\ell+2|X|$, this constitutes $n^{\Oh(\ell+|X|)}$ possible choices. 
The algorithm then treats these points as new terminals and makes a recursive call for each connected component. 
Restricting the recursion depth by $\alpha(G)$, the recursion reaches the subproblem where the graph has sufficiently high connectivity without significantly blowing up the number of paths. Overall, this yields a running time of  $n^{f(\alpha(G)+\ell)}$.

Since the \classXP algorithm exploits small separators separating highly-connected components, the first natural idea that comes to mind is to check whether the generic tools about cuts from the toolbox of parameterized algorithms could be leveraged to turn it into an \classFPT algorithm.
The recursive understanding technique introduced by Chitnis et al.~\cite{ChitnisCHPP16} is arguably the most powerful tool for designing \classFPT algorithms for such types of problems about cuts, as illustrated by the meta-theorem of Lokshtanov et al.~\cite{LokshtanovR0Z18}. Unfortunately, recursive understanding does not appear to be particularly helpful in our setting for the following reason:  in order to apply recursive understanding, one must handle the so-called \((p,q)\)-unbreakable graphs
 as the “base case.”  However, for our problem, addressing  \((p,q)\)-unbreakable graphs is not simpler than handling the general case.

Thus, we need a different method of decomposing the graph. 
To prevent the \classXP running time, the \classFPT algorithm avoids recursive calls to itself.
Instead, it encapsulates the recursion in a subroutine that in time $f(k,\ell)\cdot n^{\Oh(1)}$ either outputs an independent set of $G$ of size $k$ or finds a vertex set $S$ of size $k^{\Oh(k)}\cdot \ell$ with the following property: Each connected component of $G-S$ is $\Omega(k \cdot (\ell+|S|))$-connected. Set $S$  breaks solution paths into at most $\ell+2|S|$ parts.
 Some parts belong to $G[S]$, while others belong to $G-S$. Since the size of $S$ is in $k^{\Oh(k)}\cdot \ell$,  the paths in $G[S]$ could be enumerated in \classFPT  (in $k$ and $\ell$) time.
 However, guessing the entry and exit points of these parts, the way we describe it for the \classXP algorithm, is  too time-consuming.  
 Here is precisely the moment when the property of the decomposition comes into play. Each connected component of $G-S$ is highly-connected and thus is  Hamiltonian-$(\ell + 2|S|)$-linked. Therefore, for any selection of the paths' endpoints,  the required collection of vertex-disjoint paths spanning all component vertices always exists (unless empty). Therefore, it is sufficient for the algorithm to enumerate all possible ``patterns'' of paths' behavior between $S$ and the components of $G - S$ without enumerating the entry-exit endpoints in $G - S$.  The latter task reduces to computing a maximum matching in an auxiliary graph. This way, the whole enumeration can be done in \classFPT time.

Another complication arises because we do not assume \(\alpha(G) < k\). (Recall that we cannot compute $\alpha(G)$ in \classFPT time.)  Instead, we aim either to find a desired solution or an independent set of size \(k\). This requirement calls for a \emph{constructive} variant of the high-connectivity case. Specifically, given an \(\Omega(k \cdot \ell)\)-connected graph and \(\ell\) pairs of endpoints, we must be able to construct \(\ell\) disjoint paths connecting these endpoints and covering all vertices of the graph, or else find an independent set of size \(k\).
For this purpose, we rely on a combinatorial result by Thomas and Wollan~\cite{Thomas2005}, which states that in a \(10\ell\)-connected graph, one can always find \(\ell\) vertex-disjoint paths connecting any given pairs of endpoints. However, the proof in~\cite{Thomas2005} is non-constructive---it builds on certain properties of a minimal counterexample, and it is not obvious how to transform that argument into an algorithm that yields the required paths. 
To circumvent this issue, we make use of a Ramsey-based argument and either locate an independent set of size \(k\) or find a large clique. In the latter scenario, we can apply Menger’s theorem to connect the endpoints to the clique, thereby constructing the required \(\ell\) paths.

As a side remark, several combinatorial results---including the aforementioned work of Thomas and Wollan~\cite{Thomas2005}---guarantee the existence of \(\ell\) vertex-disjoint paths between specified pairs of endpoints in a graph with connectivity \(\Omega(\ell)\). However, there is a notable shortage of simple and efficient algorithms to \emph{compute} such vertex-disjoint paths in \(\Omega(\ell)\)-connected (or even \(\ell^{\Omega(1)}\)-connected) graphs. Our contribution in this direction may therefore be of independent interest.

Finally, for \Cref{thm:max_model}, we additionally generalize the argument above from Hamiltonian-$\ell$-linkages to arbitrary list TM-embeddings. Here, once more, the non-recursive structure of the algorithm proves helpful: With the same choice of the set $S$, we enumerate all possible behaviors of the target model on $S$ and between $S$ and the connected components of $G - S$. Again, we avoid enumerating concrete terminals and connecting vertices in $G - S$, fixing only the general connection pattern to $S$. While considerably more technical, the enumeration can still be done in \classFPT time. Then, similarly to Hamiltonian-$\ell$-linkages, it can be shown that one can always find a spanning list TM-model in a highly-connected component (or an independent set of size $k$), with arbitrary locations of the terminals. 
 
\subsection{Related Work}  
The interplay among the independence number, connectivity, Hamiltonicity, and path covers has long been a popular theme in graph theory, with historical roots going back to the works of Dilworth \cite{Dilworth1950}, Gallai and Milgram \cite{GallaiM60}, Nash-Williams \cite{NashWill71}, and Chvátal and Erdős \cite{chvatal1972note}. Surveys by Jackson and Ordaz \cite{MR1077136}, Bondy \cite{MR1373656}, and Kawarabayashi \cite{kawarabayashi2001survey} offer excellent overviews of this area. However, most of the research thus far has been centered on establishing combinatorial conditions for Hamiltonicity (or related properties), rather than  on algorithmic extensions.

\Cref{thm:above_GM} contributes to a growing line of parameterized algorithms that extend classic theorems from extremal graph theory. Examples include algorithmic adaptations of Poljak–Turzík’s theorem~\cite{MnichPSS14}, Dirac’s theorem~\cite{FominGSS21}, and Hajnal–Szemerédi’s theorem~\cite{gan2023algorithmic}. Such algorithmic extensions can be seen as part of a broader trend in parameterized complexity called \emph{above (or below) guarantee parameterization}, where one observes that
the natural parameterization of a maximization or minimization problem by the solution size is not satisfactory if there is a lower bound for the solution size that is sufficiently large.
The term \emph{above guarantee parameterization} was introduced by Mahajan, Raman, and Sikdar~\cite{MahajanRS09}, and it has proved fruitful for studying numerous fundamental problems in parameterized complexity and kernelization~\cite{AlonGKSY10,BezakovaCDF17,CrowstonJMPRS13,FominGLPSZ21,FominGSS24,FominGSS26,GargP16,DBLP:journals/mst/GutinKLM11,GutinIMY12,GutinP16,GutinRSY07,LokshtanovNRRS14,GutinMnich22}.
A feature distinguishing  \Cref{thm:above_GM} from the previous work is that it provides ``tractability below an intractable parameter''.

The exploration of long paths and cycles in graphs represents a well-established and actively pursued research direction within parameterized algorithms. Since the seminal introduction of the color-coding technique by Alon, Yuster, and Zwick \cite{AlonYZ95}, substantial progress in algorithmic methodologies has been achieved. Notably, a range of powerful algorithmic methods has been developed, including those by Bj{\"{o}}rklund \cite{DBLP:journals/siamcomp/Bjorklund14}, Koutis \cite{Koutis08}, Williams \cite{Williams09}, and Zehavi \cite{Zehavi16}. This progress has also contributed to our understanding of significantly more general and complex problems, such as topological minor embedding \cite{GroheKMW11,FominLP0Z20}.
A significant portion of parameterized algorithms deals with optimization problems on sparse graphs when specific graph parameters characterize the sparsity and structure of the graph \cite{cut-and-count,rank-treewidth,rank-treewidth2,FominLPS16,BasteST23}, see also \cite{cygan2015parameterized}. Such parameters include treewidth, treedepth, or minimum vertex cover size.  On the other hand, for more ``general'' parameters Hamiltonicity becomes intractable.   For example, the problem of deciding graph Hamiltonicity is \classW1-hard when parameterized by the clique-width of a graph \cite{FominGLS10} and is
\classParaNP-hard
parameterized by $\alpha$-treewidth (see \cite{DallardMS24} for the definition) as the problem is \classNP-complete on split graphs~\cite{golumbic2004algorithmic}.

The existence of a Hamiltonian cycle remains NP-complete on planar graphs~\cite{garey1976planar}, circle graphs~\cite{damaschke1989hamiltonian} and for several other generalizations of interval graphs~\cite{bertossi1986hamiltonian}, for chordal bipartite graphs~\cite{mu96} and for split (and therefore also for chordal) graphs~\cite{golumbic2004algorithmic}, but is solvable in linear time in interval graphs~\cite{keil1985finding} and in convex bipartite graphs~\cite{mu96}.
The existence of a Hamiltonian path can be decided in polynomial time for cocomparability graphs~\cite{damaschke1991finding} and for circular arc graphs~\cite{damaschke1993paths}. 

Many of the above mentioned graph classes are hereditary, i.e., closed in the induced subgraph order, and as such can be described by collections of forbidden induced subgraphs. This has led to carefully examining $H$-free graphs, i.e., graph classes with a single forbidden induced subgraph. For many graphs $H$, the class of $H$-free graphs has nice structural properties (e.g., $P_4$-free graphs are exactly the cographs) or their structural properties can be used to design polynomial algorithms for various graph theory problems (as a recent example, cf. colouring of $P_6$-free graphs in~\cite{Spirkl2019}).
  For Hamiltonian-type problems, consider three-vertex graphs $H$. For $H=K_3$, both Hamiltonian cycle and Hamiltonian path are NP-complete on triangle-free graphs, since they are NP-complete on bipartite graphs~\cite{mu96}. For $H=3K_1$, the edgeless graph on three vertices, both Hamiltonian cycle and Hamiltonian path are polynomial time solvable~\cite{Duf1981}. The remaining two graphs, $P_3$ and $K_1+K_2$, are induced subgraphs of $P_4$, thus the corresponding class of $H$-free graphs is a subclass of cographs, and as such a subclass of cocomparability graphs, in which the Hamiltonian path and Hamiltonian cycle problems are solvable in polynomial time~\cite{damaschke1991finding}. However, a complete characterization of graphs $H$ for which Hamiltonian path or Hamiltonian cycle problems are solvable in polynomial time (and for which they are NP-complete) is not in sight. Until now, the complexity has been open even for the next smallest edgeless graph $4K_1$. \Cref{thm:max_model}  answers this complexity question for all edgeless forbidden induced subgraphs.
In a companion paper 
Jedli\v{c}kov{\'{a}} and   Kratochv{\'{\i}}l \cite{jedlivckova2024structure} explicitly described situations
for $\alpha(G)=3,4$ and $5$ in which a Hamiltonian path does exist.

The parameterized complexity of many graph problems remains insufficiently explored for parameters capturing a graph’s density—such as its independence number. It is surprising that, despite the substantial progress in understanding Hamiltonicity, connectivity, and topological minors on graphs of bounded treewidth (which are sparse), relatively little has been achieved for graphs of bounded independence number (which are typically dense). 
Only a handful of results (\classFPT or \classXP) address parameterization by a graph’s independence number. These include the \classXP algorithm of Ovetsky~Fradkin  and Seymour \cite{FradkinS15} for the edge-disjoint paths problem in directed graphs, the kernelization algorithms of Lochet et al.~\cite{LochetLM0SZ20} for various cycle-hitting problems in directed graphs, the subexponential parameterized algorithms of Misra et al.~\cite{MisraSSZ23} for digraphs with constant independence number, and the recent work of El Maalouly, Steiner, and Wulf~\cite{MaaloulySW23} on \classFPT algorithms for exact matching.

\section{Preliminaries}\label{sec:preliminaries}

We follow the standard graph-theoretic notation and refer to the textbook Diestel~\cite{Diestel} for basic notions.   
We consider simple undirected graphs without loops or multiple edges.  
The vertex set of a graph $G$ is denoted by $V(G)$, its edge set by $E(G)$. Edges are considered as two-element sets of vertices, thus we write $u\in e$ to express that a vertex $u$ is incident to an edge $e$.  For the sake of brevity, we write $uv$ instead of $\{u,v\}$ for the edge containing vertices $u$ and $v$. We say that $u$ is {\em adjacent to} $v$ if $uv\in E(G)$. 
The \emph{open neighborhood} of a vertex $v$ is the set $N_G(v)$ of the vertices adjacent to $v$ in $G$.
The {\em degree} $\deg_G(v)$ of a vertex $v$ is the number of other vertices adjacent to it. The subgraph of $G$ induced by vertices $S\subseteq V(G)$ will be denoted by $G[S]$. We use $G-S$ for $G[V(G)\setminus S]$.

A set $I\subseteq V(G)$ of pairwise non-adjacent vertices is said to be \emph{independent}.  
The {\em independence number}  of a graph $G$, denoted by $\alpha(G)$, is the order of the largest independent set of $G$. With the standard notion of $K_k$ being the complete graph with $k$ vertices and $G+H$ being the disjoint union of graphs $G$ and $H$, $\alpha(G)$ is equal to the largest $k$ such that $kK_1$ is an induced subgraph of $G$.  A graph is called {\em $H$-free} if it contains no induced subgraph isomorphic to $H$.  An edgeless graph is said to be \emph{empty}.

A {\em path} $P=v_0,\ldots,v_s$ in a graph $G$ is a sequence of distinct vertices such that any two consecutive ones are adjacent; the vertices $v_0$ and $v_s$ are \emph{endpoints}, and the other vertices are internal. For $s,t\in V(G)$, we say that $P$ is an $s$-$t$-path if $s$ and $t$ are its endpoints. Although we consider undirected graphs, when we talk about a path in a graph, the path itself is considered traversed in the direction from its starting vertex to the ending one. Formally, when we say that a path $P$ connects a vertex $s$ to a vertex $t$, then by $P^{-1}$ we denote the same path, but traversed from $t$ to $s$. This is important when creating a longer path by concatenating shorter ones.  
The {\em length} of a path is the number of its edges. A {\em cycle} is formed by a path of length at least two by adding an edge connecting the endpoints. The path (cycle) is {\em Hamiltonian} if it contains all vertices of the graph. The \emph{path cover number} of $G$, denoted by $\mathrm{pc}(G)$, is the smallest number of vertex disjoint paths that cover all vertices of $G$.

A graph is connected if any two vertices are connected by a path.
A {\em vertex  cut} in a graph $G$ is a set $A\subset V(G)$ of vertices such that the graph $G-A=G[V(G)\setminus A]$ is disconnected. The {\em vertex connectivity} $c_v(G)$ of a graph $G$ is the order of a minimum cut in $G$, or $|V(G)|-1$ if $G$ is a complete graph. Since we will not consider edge connectivity, we will often omit the adjective when talking about the connectivity measure; we always have vertex connectivity in mind.

We build upon the following results of Chv\'atal and Erd\H{o}s which we will later generalize to Hamiltonian linkages in \Cref{prop:comb-highly-connected}. 

\begin{proposition} [Chv\'{a}tal and Erd\H{o}s \cite{chvatal1972note}] \label{prop:chvatal} 
Let $G$ be an $s$-connected graph with at least three vertices. 
\begin{enumerate}
\item If $\alpha(G)<s+2$, then $G$ has a Hamiltonian path. 
\item If $\alpha(G)<s+1$, then $G$ has a Hamiltonian cycle.
\item If $\alpha(G)<s$, then $G$ is Hamiltonian connected (i.e., every pair of vertices is joined by a Hamiltonian path).
\end{enumerate}
\end{proposition}

\section{Proof of \Cref{thm:above_GM}: Path cover below Gallai-Milgram}\label{sec:GM}
In this section, we prove~\Cref{thm:above_GM}. 
Recall that in the \textsc{Path Cover} problem, we are given a graph $G$ and an integer $p\ge 1$, and the task is to decide whether it is possible to cover $G$ with at most $p$ vertex-disjoint paths.
Throughout this section, by $\operatorname{pc}(G)$, we denote the minimum number of vertex-disjoint paths required to cover $G$.
A well-known theorem by Gallai and Milgram  \cite{GallaiM60} guarantees that $\operatorname{pc}(G)\leq \alpha(G)$. When we combine this theorem with \Cref{thm:max_model}, we obtain an algorithm solving \textsc{Path Cover} in time $2^{\alpha(G)^{\Oh(\alpha(G)^2)}}\cdot |G|^{\Oh(1)}$. 
Our contribution in \Cref{thm:above_GM} takes a further step by establishing the tractability of \textsc{Path Cover} for parameterization $k=\alpha(G)-p$, that is ``below'' the bound of the Gallai-Milgram theorem. 
We restate \Cref{thm:above_GM} here for the reader's convenience.

\thmabovgm*

We require two auxiliary results to prove \Cref{thm:above_GM}.
\subsection{Small separators: Reducing connected components}
The first result is dealing with small vertex separators of $G$. It provides a useful tool that 
for a small vertex separator $S$ of $G$,  reduces the number of connected components in $G-S$ and the value of $\operatorname{pc}(G)$ simultaneously.

\begin{lemma}\label{lemma:selected_cliques}
	Let $G$ be a graph,  $S\subseteq V(G)$, and let  $C_1, C_2, \ldots, C_t\subseteq V(G)\setminus S$ be disjoint sets such that for each $i\in[t]$, $C_i$ is a clique of $G$ inducing a connected  component of $G-S$.
	Then there is a polynomial time algorithm that, given $G$, $S$, and $\{C_i\}_{i\in[t]}$, outputs a set of at most $2|S|^2$ indices $X\subseteq [t]$ such that $G$ has a  minimum path cover containing a path with vertex set $C_i$ for each $i\in[t]\setminus X$.
\end{lemma}

\begin{proof}
	The construction of $X$ is based on the following observation. We say that a path $P$ in a path cover $\mathcal{P}$ of $G$ is \emph{degenerate} if $V(P)=C_i$ for some $i\in [t]$.
	
	\begin{claim}\label{claim:few_used_cliques}
		In any minimum path cover $\mathcal{P}$ of $G$, at most $2|S|$ of the sets $C_1, C_2, \ldots, C_t$ contain vertices of non-degenerate paths of $\mathcal{P}$.
	\end{claim}
	
	\begin{claimproof}
		Notice that if $\mathcal{P}$ has no path with an edge $uv$ with $u\in S$ and $v\in C_i$ for some $i\in[t]$ then the clique $C_i$ is covered by a degenerate path by the minimality of $\mathcal{P}$. 
		Hence, if $C_i$ contains a vertex of a non-degenerate path then $\mathcal{P}$ has a path $P$ with an edge with one endpoint in $S$ and the other in $C_i$. 
		Each vertex in $S$ is incident to at most two edges in the paths from $\mathcal{P}$. Thus, at most $2|S|$ edges are in the paths of $\mathcal{P}$ with one endpoint in $S$ and the other in one of $C_1, C_2,\ldots, C_t$.
		Hence, at most $2|S|$ of the sets can contain vertices of non-degenerate paths.
	\end{claimproof}
	
	To construct $X$, for each vertex $v\in S$, we \emph{mark} at most $2|S|$ cliques from $\{C_i\}_{i\in[t]}$ that contain neighbors of $v$. If there are at most $2|S|$ such cliques, then we mark all of them. Otherwise, we mark arbitrary $2|S|$ cliques containing neighbors of $v$. 
	Then $X$ is defined as the set of indices of the marked cliques. It is straightforward that $|X|\leq 2|S|^2$ and $X$ can be constructed in polynomial time.
	We claim that there is a minimum path cover $\mathcal{P}$ of $G$ such that for each $i\in[t]\setminus X$, there is a path $P\in \mathcal{P}$ satisfying $V(P)=C_i$.  In other words,  each unmarked clique is covered by a degenerate path. 
	
	Given a path cover $\mathcal{P}$ of $G$, we say that an edge $uv$ of a path $P\in\mathcal{P}$ is \emph{bad} if $u\in S$ and $v$ is a vertex of an unmarked clique. We show that there is a minimum path cover $\mathcal{P}$ of $G$ that does not contain paths with bad edges. The proof is by contradiction. We choose $\mathcal{P}$ as a minimum path cover of $G$ that minimizes the total number of bad edges in the paths. Suppose that there is $P\in\mathcal{P}$ that contains a bad edge $uv$ with $u\in S$ and $v\in C_i$ for some $i\in[t]$. We have that $C_i$ is unmarked and has a neighbor of $u$. Then, by the definition of the marking procedure, there are at least $2|S|$ marked cliques with neighbors of $u$. Because $C_i$ contains a vertex of the non-degenerate path $P$, we obtain that there is a marked clique $C_j$ for some $j\in[t]$ with a neighbor $w$ of $u$ such that $C_j$ is covered by a degenerate path $P'\in \mathcal{P}$ by \Cref{claim:few_used_cliques}. Because $C_j$ is a clique, we can assume without loss of generality that $w$ is an end-vertex of $P'$. Deleting the edge $uv$ separates $P$ into two paths $Q$ and $Q'$ having end-vertices $v$ and $u$, respectively. We use the fact that $uw\in E(G)$ and construct the path $\hat{Q}$ from $Q'$ and $P'$ by joining their end-vertices $u$ and $w$ by the edge $uw$. Notice that $Q$ and $\hat{Q}$ are vertex-disjoint and $V(Q)\cup V(\hat{Q})=V(P)\cup V(P')$. This allows us to construct the path cover $\mathcal{P}'$ by replacing $P$ and $P'$ in $\mathcal{P}$ by $Q$ and $\hat{Q}$.  Because the number of paths in $\mathcal{P}'$ is the same as in $\mathcal{P}$, $\mathcal{P}'$ is a minimum path cover. However, the number of bad edges in the paths of $\mathcal{P}'$ is less than the number of bad edges in the paths of $\mathcal{P}$, contradicting the choice of $\mathcal{P}$. This proves that there is a minimum path cover $\mathcal{P}$ of $G$ that does not contain paths with bad edges. 
	
	To complete the proof, recall that if a minimum path cover has no path with an edge $uv$ with $u\in S$ and $v\in C_i$ for some $i\in[t]$ then the clique $C_i$ is covered by a degenerate path. Thus, there is a minimum path cover $\mathcal{P}$ of $G$ such that for each $i\in[t]\setminus X$, there is a path $P\in \mathcal{P}$ satisfying $V(P)=C_i$. 
\end{proof}

\subsection{The case of $\operatorname{pc}(G)<k^C$}
The second auxiliary result is a ``lighter version'' of \Cref{thm:above_GM}.
It deals with the particular case of $\operatorname{pc}(G)<k^C$.

\begin{lemma}\label{lemma:solve_given_small_cover}
	There is an algorithm with running time $2^{k^{\Oh(k^{2C})}}\cdot \polyn$ that, given an $n$-vertex graph $G$, integer $k\ge 0$ and a path cover of $G$ with at most $\Oh(k^C)$ paths for constant $C\ge 1$, outputs a path cover $\mathcal{P}$ of $G$ and
	 \begin{itemize}
		\item either correctly reports that  $\mathcal{P}$ is a minimum-size path cover, or 
		\item outputs together with $\mathcal{P}$ an independent set of size $|\mathcal{P}|+k$ certifying that  $\mathcal{P}$ contains at most $\alpha(G)-k$ paths.	   
	\end{itemize}

\end{lemma}
\begin{proof}
	We present an algorithm $\mathcal{A}$.
	The algorithm assumes that $G$ 
	contains at least one edge as, otherwise, $pc(G)=n$ and the minimum-size path cover consists of trivial paths. 
	
	Given $G,k$ and a path cover $\mathcal{P'}$ of $G$ of size $p'$, the algorithm $\mathcal{A}$ tries to find an optimal path cover of $G$.
	To achieve that, the algorithm iterates an integer $i$ from $1$ up to $\min\{p',n-1\}$ and an integer $j$ from $0$ to $i-1$.
	On each iteration, the algorithm constructs a graph $H_{ij}$ defined as the disjoint union of $j$ copies of $K_1$ and $i-j$ copies of $K_2$, that is, $H_{ij}$ contains $j$ isolated vertices and $i-j\geq 1$ edges forming a matching. Notice that 
	$G$ admits a path cover with $i$ paths if and only if $H_{ij}$ has a spanning TM-model in $G$ for some $j\in\{0,\ldots,i-1\}$ because $G$ has at least one edge and, therefore, at least one path in a path cover of minimum size is nontrivial.
	The algorithm $\mathcal{A}$ thus applies the algorithm of \Cref{thm:max_model} as a subroutine to $H_{ij}, G$ and trivial list assignment $L\equiv V(G)$, with integer parameter set to $p'+k=\Oh(k^C)$.
	The subroutine works in $$2^{|H_{ij}|^{\Oh(p'+k)}\cdot(p'+k)^{\Oh((p'+k)^2)}}\cdot \polyn=2^{k^{\Oh(k^{2C})}}\cdot\polyn$$ running time.
	
	If the outcome of the subroutine is an independent set $I$ in $G$ of size $p'+k$, then the algorithm outputs $\mathcal{P}'$ and $I$ as the solution.
	Since $|\mathcal{P}'|= p'$, $|I|\ge |\mathcal{P}'|+k$ is satisfied, so this solution satisfies the lemma statement.
	
	If the subroutine decides that $H_{ij}$ does not have a TM-model in $G$ or the maximum TM-model of $H_{ij}$ in $G$ is not spanning, then $G$ cannot be covered with $i$ paths exactly $j$ of which are trivial and $\mathcal{A}$ continues to the next iteration.
	If the subroutine gives a spanning TM-model of $H_{ij}$ in $G$, then for the first such value $i$, 
    the minimum path cover size of $G$ equals $i$ and $\mathcal{A}$ transforms this TM-model into a path cover of $G$ with $i$ paths.
	Then $\mathcal{A}$ outputs this optimal path cover of $G$ and stops.
	
	Note that $\mathcal{A}$ necessarily stops and produces a correct solution at some iteration, as it iterates $i$ over $\{1,2,\ldots,\min\{p',n-1\}\}$ while $1\le \operatorname{pc}(G) \le \min\{p',n-1\}$.
	The proof of the lemma is complete.
\end{proof}

\subsection{Proof of \Cref{thm:above_GM}}
We are ready to prove \Cref{thm:above_GM}.
	We describe an algorithm that starts from a trivial path cover of $G$ and then gradually transforms it.
	These transformations are indicated as \emph{reduction rules} and are performed in polynomial time.
	We show that each such rule does not increase the number of paths in the path cover.
	Moreover, each reduction rule necessarily reduces either the total number of paths or the number of \emph{special} (to be defined later) paths in the path cover.
	
	A series of consecutive applications of the reduction rules always leads the algorithm to cases when the algorithm can produce a solution satisfying the theorem statement in $f(k)\cdot \polyn$ time and stop.
	We indicate such cases and the corresponding behavior of the algorithm as \emph{solution subroutines}.
    Thus, every run of the algorithm ends with the solution subroutine run.
    
	In what follows, we present the \emph{reduction rules} and \emph{solution subroutines}. Each of the reduction rules reduces the number of paths in a path cover or reduces the number of some special paths. Solution subroutines output path covers and independent sets.  The rules and subroutines are given in the order in which the algorithm checks their applicability.
	That is, the algorithm applies a reduction rule or solution subroutine only if no reduction rule or solution subroutine before it can be applied.

	When we reach the situation where the obtained path cover $\mathcal{P}$ and graph $G$ are irreducible, we deduce that the instance has nice structural properties. In this case, the final solution subroutine calls the algorithms of \Cref{lemma:selected_cliques,lemma:solve_given_small_cover}.

	The algorithm starts with initializing the path cover $\mathcal{P}$ with $|V(G)|$ paths of length $0$.
	Suppose that at some stage of the algorithm, we have constructed a path cover with $m$ paths for some $m\geq 1$. 
	Let the current set of paths be $\mathcal{P} = \{P_1, P_2, \ldots, P_m\}$.
	We denote the endpoints of $P_i$ by $s_i$ and $t_i$.

	The first reduction rule  is straightforward and comes from the proof of the Gallai-Milgram theorem for undirected graphs.
	
	\begin{rrule}\label{rrule:distinct_path_endpoints}
		If there is an edge $uv$ in $G$ such that  $u \in\{s_i,t_i\}$ and $v\in\{s_j,t_j\}$ and $i\neq j$,
		then join $P_i$ and $P_j$ into a single path via $uv$.
	\end{rrule}
	
	The correctness of \Cref{rrule:distinct_path_endpoints} is trivial, and it reduces the number of paths in $\mathcal{P}$ by one.
	Clearly, the exhaustive application of \Cref{rrule:distinct_path_endpoints} to $\mathcal{P}$ yields $m\le \alpha(G)$, as $\{s_1, s_2, \ldots, s_m\}$ induces an independent set in $G$.
	When \Cref{rrule:distinct_path_endpoints} cannot be applied, then the only edges in $G$ allowed between the endpoints of paths in $\mathcal{P}$ are the edges of form $s_it_i$.
	
	\medskip\noindent\textbf{Special and usual paths.}
	We call a path $P_i$ in the path cover \emph{special}, if either $s_i=t_i$ or its endpoints are adjacent; note that if $P_i$ consist of the single edge $s_it_i$, then it is special.
	If a path $P_i$ is not special, we call it \emph{usual}.
	We claim that sufficiently many usual paths give a solution to the problem.
	
	\begin{claim}\label{claim:many_usuals_is_win}
		If at least $k$ paths are usual in the path cover   $\mathcal{P}=P_1, \ldots, P_m$ and \Cref{rrule:distinct_path_endpoints} is not applicable, then $m\le \alpha(G)-k$. Moreover, an independent set $I$  in $G$ of size  $m+k$ is computable in polynomial time.
	\end{claim}
	\begin{claimproof}
		We show that $G$ has an independent set $I$ of size $m+k$.
		Initialize $I$ with $\{s_1, s_2,\ldots, s_m\}$.
		Since \Cref{rrule:distinct_path_endpoints} is not applicable, $I$ is independent in $G$.
		
		Let $i_1, i_2, \ldots, i_k$ be $k$ distinct indices of usual paths.
		Since each $P_{i_j}$ is usual, $t_{i_j}\neq s_{i_j}$ and $s_{i_j}$ is not adjacent to $t_{i_j}$.
		Also, \Cref{rrule:distinct_path_endpoints} is not applicable, so there is no edge between $t_{i_j}$ and any vertex in $I$, while $\{t_{i_1}, t_{i_2}, \ldots, t_{i_k}\}$ is also an independent set in $G$.
		
		Put $t_{i_1}, t_{i_2}, \ldots, t_{i_k}$ in $I$.
		We obtain an independent set of size $m+k$ in $G$, consequently $\alpha(G)-k\ge m$.
	\end{claimproof}
	
	\Cref{claim:many_usuals_is_win} yields the corresponding solution subroutine for the algorithm.
	
	\begin{trule}\label{trule:too_many_usuals}
		If there are at least $k$ usual paths among $\mathcal{P}=P_1, P_2,\ldots, P_m$, output $\mathcal{P}$ and $I$ as in \Cref{claim:many_usuals_is_win}.
	\end{trule}
	
	From now on,  we assume that $\mathcal{P}$ has fewer than $k$ usual paths.
	The remaining reduction rules deal primarily with special paths.
	The following rule reduces the total number of paths given that $G$ has an edge between \emph{any} vertex of a special path and an endpoint of an arbitrary other path.

	\begin{rrule}\label{rrule:special_with_usual}
		If $P_i$ is a special path of length at least two and there is an edge in $G$ connecting some $u\in V(P_i)$ with $s_j$ (or $t_j$) for some $j\neq i$, then
		\begin{enumerate}
			\item Add edge $s_it_i$ to $P_i$, so it becomes a cycle  $C_i$;
			\item Remove edge $uv$ from $C_i$, where $v$ is any neighbor of $u$ in $C_i$, so it becomes a $u$-$v$-path;
			\item Join $P_i$ and $P_j$ via $us_j$ (or $ut_j$).
		\end{enumerate}
	\end{rrule}
	
	\Cref{rrule:special_with_usual} replaces $P_i$ and $P_j$ with a path covering $V(P_i)\cup V(P_j)$, so $\mathcal{P}$ remains a path cover implying that the rule is correct.
	The following reduction rule deals with an edge between two vertices of distinct special paths.
	It is demonstrated in \Cref{fig:rrules}a.
	
	\begin{rrule}\label{rrule:merge_special_paths}
		If $P_i$ and $P_j$ are special paths of length at least two for $i\neq j$, and there is an edge in $G$ connecting $u_i \in V(P_i)$ and $u_j \in V(P_j)$, then
		\begin{enumerate}
			\item Add edge $s_jt_j$ to $P_j$, so it becomes a cycle $C_j$;
			\item Remove edge $u_jv_j$ from $C_j$ for a neighbor $v_j$ of $u_j$ in $C_j$, so it becomes a $u_j$-$v_j$-path.
			\item Apply \Cref{rrule:special_with_usual} to $P_i$, $P_j$ and $u_i u_j$.
		\end{enumerate}
	\end{rrule}

	Again, \Cref{rrule:merge_special_paths} replaces $P_i$ and $P_j$ with a path covering  $V(P_i)\cup V(P_j)$, thus reducing the number of paths in $\mathcal{P}$.
	We claim the following.
	\begin{claim}\label{claim:no_edges_between_specials}
		If Reduction Rules \ref{rrule:distinct_path_endpoints}, \ref{rrule:special_with_usual}, and \ref{rrule:merge_special_paths} are not applicable, then for each pair of distinct special paths $P_i, P_j$ in $\mathcal{P}$, there is no edge between $V(P_i)$ and $V(P_j)$ in $G$.
	\end{claim}
	\begin{claimproof}
		Assume that the reduction rules are not applicable, but there are special paths $P_i, P_j \in \mathcal{P}$ with $i\neq j$ with an edge between $u_i \in V(P_i)$ and $u_j \in V(P_j)$.
		If both $P_i$ and $P_j$ consist of more than two vertices, then \Cref{rrule:merge_special_paths} is applicable.
		Thus, without loss of generality, $P_j$ consists of one or two vertices, so $V(P_j)=\{s_j, t_j\}$ and $u_j\in\{s_j,t_j\}$.
		If $P_i$ consists of more than two vertices, then \Cref{rrule:special_with_usual} is applicable, which is a contradiction.
		Then, the length of $P_i$ is less than two, implying that $u_i\in\{s_i,t_i\}$.
		In this case, \Cref{rrule:distinct_path_endpoints} is applicable, leading us to the final contradiction. 
	\end{claimproof}
	
	The claim is the basis for the next solution subroutine, which the algorithm applies next.
	We explain how it works by proving the following claim.
	
	\begin{claim}\label{claim:too_many_non_cliques}
		If \Cref{trule:too_many_usuals} and Reduction Rules~\ref{rrule:distinct_path_endpoints}, \ref{rrule:special_with_usual}, and \ref{rrule:merge_special_paths} are not applicable, and there are at least $2k$ special paths $P_i$ such that $V(P_i)$ is not a clique in $G$, then $m< \alpha(G)-k$.  Moreover, an independent set $I$  in $G$ of size more than $m+k$ is computable in polynomial time.

	\end{claim}
	\begin{claimproof}
		We construct an independent set $I$ consisting of more than $m+k$ vertices.
		Since \Cref{trule:too_many_usuals} is not applicable, there are more than $m-k$ special paths.
		For each special path $P_i$ such that $V(P_i)$ is not a clique, put two independent vertices of $V(P_i)$ in $I$.
		For each other special path, put its arbitrary vertex into $I$.
		
		By \Cref{claim:no_edges_between_specials}, $I$ is an independent set in $G$.
		Since there are at least $2k$ non-cliques among vertex sets of special paths, $|I|>(m-k)+2k=m+k$.
	\end{claimproof}
	
	The solution subroutine based on \Cref{claim:too_many_non_cliques} is given below.
	The algorithm applies it when none of the reduction rules and solution subroutines above are applicable.
	
	\begin{trule}\label{trule:too_many_non_cliques}
	If there are at least $2k$ special paths $P_i$ such that $V(P_i)$ is not a clique in $G$, then output $\mathcal{P}$ and an $I$ as in   \Cref{claim:too_many_non_cliques}.
	\end{trule}

	The algorithm then exploits the edges between special and usual paths in $\mathcal{P}$.
	The following reduction rule reduces the number of special paths in $\mathcal{P}$ while not increasing the total number of paths in $\mathcal{P}$.

\input{path-cover-transformations-fig}

	\begin{rrule}\label{rrule:two_specials_to_one}
		If $P_i$ and $P_j$ are special paths for $i\neq j$, and there is a usual path $P_\ell$, and there are two edges $u_ix_i, u_jx_j$ in $G$ with $u_i\in V(P_i), u_j\in V(P_j), x_i,x_j\in V(P_\ell)$, $x_i\neq x_j$, and $x_i$ goes before $x_j$ in $P_\ell$, then
		\begin{enumerate}
			\item 
            If $u_i$ is not an endpoint of $P_i$ or $u_j$ is not an endpoint of $P_j$, 
            make the vertices endpoints of the corresponding paths similarly to the first two steps of \Cref{rrule:special_with_usual};
			\item Break $P_\ell$ into three (or two) parts: the prefix part $s_\ell P_\ell x_i$, the suffix part $x_j P_\ell t_\ell$, and, if it exists, the remaining middle part;
			\item Join $P_i$ with the prefix part via $u_ix_i$;
			\item Join $P_j$ with the suffix part via $u_jx_j$.			
		\end{enumerate}
	\end{rrule}

	\Cref{fig:rrules}b illustrates the result of application of \Cref{rrule:two_specials_to_one}.
	The following claim addresses its correctness.

	\begin{claim}
		If Reduction Rules \ref{rrule:distinct_path_endpoints}, \ref{rrule:special_with_usual} are not applicable, then \Cref{rrule:two_specials_to_one} reduces the number of special paths in $\mathcal{P}$ and does not increase the number of paths in $\mathcal{P}$.
	\end{claim}
	\begin{claimproof}
		The second part of the claim is trivial, since \Cref{rrule:two_specials_to_one} produces three (or two, if $x_ix_j \in E(P_\ell)$) paths by transforming three distinct paths $P_i, P_j, P_\ell$, so $|\mathcal{P}|$ cannot increase. 
		
		To prove the first part, denote by $v_i$ and $v_j$ the endpoints of the modified $P_i$ and $P_j$ distinct from $u_i$ and $u_j$, respectively (see~\Cref{fig:rrules}b). 
        We show that the resulting $v_i$-$s_\ell$-path $Q_i$ and $v_j$-$t_\ell$-path $Q_j$ obtained by the rule are usual.
		By symmetry, it is sufficient to give the proof for $Q_i$.
    Suppose that $Q_i$ is special, that is, $v_is_\ell \in E(G)$. However, then we have an edge between a special path $P_i$ and an endpoint of $P_\ell$ implying that either \Cref{rrule:distinct_path_endpoints} or \Cref{rrule:special_with_usual} is applicable. This contradiction proves that $Q_i$ is usual.
    Since $Q_i$ and $Q_j$ are usual, we have that \Cref{rrule:two_specials_to_one} reduces the number of special paths. This completes the proof.
	\end{claimproof}

	The list of reduction rules of the algorithm is exhausted.
	Starting from this point, the algorithm exploits the structure of the graph, assuming that no reduction rule is applicable.
	In what follows, we explain these structural properties. We also give the last solution subroutine of the algorithm.
	
	\begin{claim}\label{claim:special_separator}
		If none of the Reduction Rules \ref{rrule:distinct_path_endpoints}-\ref{rrule:two_specials_to_one} and Solution Subroutines 1,2 can be applied to $\mathcal{P}$, then there exists $S\subseteq V(G)$ with $|S|<k$ such that $V(P_i)$ is a connected component in $G-S$ for at least $m-2k$ special paths $P_i$. Moreover, such a set $S$ could be constructed in polynomial time.
	\end{claim}
	\begin{claimproof}
		Since
        Solution Subroutine~1
        is not applicable, there are fewer than $k$ usual paths in $\mathcal{P}$.
		We construct $S$ by taking one or zero vertices from each usual path, so $|S|<k$ holds automatically.
		We now explain the choice of the vertex to put in $S$ for each usual path.
		
		Consider a usual path $P_i\in \mathcal{P}$.
		We call a vertex $u\in V(P_i)$ a \emph{connector} if there is a special path $P_j$ and a vertex $v \in V(P_j)$ with $uv \in E(G)$.
		If $P_i$ has exactly one connector, then we put this connector in $S$.		
		Otherwise, i.e.\ when $P_i$ has zero  connectors or more than one connector, we do not put any vertex of $P_i$ in $S$.
		
		We show that at least $m-2k$ special paths are isolated in $G-S$.
		First note that there are at least $m-k$ special paths in $\mathcal{P}$, since the number of usual paths is bounded.
		
		Assume now that $P_j$ is a special path that is not isolated in $G-S$.
		Then there is an edge $uv$ with $u \in V(P_i)$, $v \in V(P_j)$, $i\neq j$ and $u\notin S$.
		By \Cref{claim:no_edges_between_specials}, $P_i$ can only be usual.
		Since $u$ is a connector of $P_i$ and $u\notin S$, $P_i$ has two or more distinct connectors.
		As \Cref{rrule:two_specials_to_one} is not applicable, all these connectors can only connect $P_i$ with $P_j$.
		Hence, each usual path $P_i$ gives at most one non-isolated special path in $G-S$.
		Thus, at most $k$ special paths in $G-S$ are not isolated, implying that at least $(m-k)-k=m-2k$ special paths are isolated in $G-S$.
	\end{claimproof}

	We are ready to present the final solution subroutine which combines all algorithmic results given previously in this section.

	\begin{trule}\label{trule:final} Perform the following steps:
		\begin{enumerate}
			\item Obtain $S$ according to \Cref{claim:special_separator} and put $h:=\max\{0, m-4k\}$;
			\item Order the paths in $\mathcal{P}$ so that, for each $i \in [h]$,  the vertices  $V(P_i)$ induce a complete connected component in $G-S$;
			\item Apply the algorithm of \Cref{lemma:selected_cliques} to $G$, $S$ and $V(P_1), V(P_2), \ldots, V(P_{h})$ and obtain a set $X$ of at most $2k^2$ indices, where each $i \in X$ satisfies $i \in [h]$;
			\item Obtain $G'$ and $\mathcal{P}'$ by removing $V(P_j)$ from $G$ and $P_j$ from $\mathcal{P}$, for each $j$ such that $j\in [h]$ and $j\notin X$;
			\item Apply the algorithm of \Cref{lemma:solve_given_small_cover} as a subroutine to $G',\mathcal{P}'$ and $k':=2k$ and obtain a path cover $\mathcal{S}'$ of $G'$ and (possibly) an independent set $I'$ in $G'$;
			\item Obtain a path cover $\mathcal{S}$ of $G$ by adding all paths in $\mathcal{P}\setminus\mathcal{P}'$ to $\mathcal{S}'$;
			\item If $I'$ was obtained, remove all vertices of $S$ from $I'$ and for each path $P_j\in\mathcal{P}\setminus\mathcal{P}'$, add one vertex of $P_j$ to $I'$. Denote the resulting set by $I$;
			\item Output $\mathcal{S}$ and (if obtained) $I$.
		\end{enumerate}
	\end{trule}

	We conclude the proof with a claim certifying the correctness of the solution subroutine.
	In contrast with all previous rules and subroutines, \Cref{trule:final} requires superpolynomial (but \classFPT in $k$) running time.
	
	\begin{claim}
		If Reduction Rules \ref{rrule:distinct_path_endpoints}-\ref{rrule:two_specials_to_one} and Solution Subroutines 1,2 are not applicable, then
		\Cref{trule:final} works in $2^{k^{\Oh(k^4)}}\cdot |G|^{\Oh(1)}$ time and produces a feasible solution.
	\end{claim}
	\begin{claimproof}
		First note that at least $m-4k$ complete connected components are guaranteed in $G-S$ by \Cref{claim:special_separator} ($S$ separates $m-2k$ special paths) and inapplicability of \Cref{trule:too_many_non_cliques} (at most $2k$ vertex sets of special paths are not cliques).

	We now move on to establish the upper bound on the running time.
All steps, except for applying \Cref{lemma:solve_given_small_cover}, are carried out in polynomial time. We have
\[
|\mathcal{P}'|=|\mathcal{P}|-h+|X|=m-h+|X|\le 4k+2k^2=\Oh(k^2).\]

		By applying  \Cref{lemma:solve_given_small_cover}  to $G',k'$ and $\mathcal{P}'$, we  obtain the running time  $2^{k^{\Oh(k^4)}}\cdot |G|^{\Oh(1)}$.
		
	Set $\mathcal{S}$ is obtained from a path cover $\mathcal{S'}$ of $G'$ via adding paths covering $V(G)\setminus V(G')$. Hence it is also a path cover. 
		On the other hand, $I$ is (if obtained) an independent set in $G$ because $I\subseteq V(G-S)$ and $I$ is independent in $G-S$.
		
		Now we prove that either $\mathcal{S}$ is optimal or $|I|-|\mathcal{S}|\ge k$.
		First, let us consider the scenario where $I$ is not produced. This can occur only when 
		 the algorithm of \Cref{lemma:solve_given_small_cover} does not produce $I'$.
		Then \Cref{lemma:solve_given_small_cover} guarantees that $\mathcal{S'}$ is an optimal path cover of $G'$ and $|\mathcal{S}'|=\operatorname{pc}(G')$.
		By \Cref{lemma:selected_cliques}, the removal of $V(P_j)$ for all $j\in [h]\setminus X$ decreases $\operatorname{pc}(G)$ by $h-|X|$, 
        so $\operatorname{pc}(G')=\operatorname{pc}(G)-(h-|X|)$.
		Since $|\mathcal{P}\setminus \mathcal{P}'|=h-|X|$, we have $|\mathcal{S}|=\operatorname{pc}(G)$. It means that $\mathcal{S}$ is an optimal path cover of $G$.

Finally, let us consider the situation where \(I\) is produced. In this case, the algorithm in \Cref{lemma:solve_given_small_cover} provides 
\(\mathcal{S'}\) and \(I'\) with \(|I'| - |\mathcal{S}'| = k' = 2k\). We notice that

\[ |I|\ge |I'|-|S|+|\mathcal{P}\setminus \mathcal{P}'|> |I'|-k+(h-|X|).\]

		As discussed in the previous paragraph, $|\mathcal{S}|-|\mathcal{S}'|=h-|X|$.
		We conclude that  \[|I|-|\mathcal{S}| =2k+ (|I|-|I'|)-(|\mathcal{S}|-|\mathcal{S}'|)>2k+(-k+h-|X|)-(h-|X|)=k.\]
		The feasibility of $\mathcal{S}$ and $I$ is proved.
	\end{claimproof}

	This completes the proof of  \Cref{thm:above_GM}.

\section{Proof of  \Cref{thm:max_model}: List TM-embedding parameterized by independence number}

In this section, we prove \Cref{thm:max_model}. The proof consists of three main steps, deferred to their respective subsections. The first step is the proof of the stronger version of the theorem for highly-connected graphs which we call Spanning Lemma. It shows that for any fixed injection $f: V(H) \to V(G)$ the desired TM-embedding respecting $f$ or an independent set can be found; moreover, the found TM-embedding spans $V(G)$. The second step is the so-called Merging Lemma, that combines the solutions from several highly-connected components, provided that the size of the separator between them is small. Finally, to obtain the proof of \Cref{thm:max_model} itself, we show an iterative procedure that either finds the desired separator, or an independent set.

\subsection{Spanning Lemma}

We first deal with highly-connected graphs. More precisely, we demonstrate that, given an injective mapping $f: V(H) \to V(G)$, we can either compute a TM-embedding of $H$ in $G$ respecting $f$ and spanning $V(G)$, or find an independent set of size $k$ in $G$; this holds if the connectivity of $G$ is sufficiently high.
 Recall that a graph $G$ is \emph{$\ell$-linked} if it has at least $2\ell$ vertices and for every set of $\ell$ disjoint pairs $(s_j, t_j)$ of distinct vertices, there exist $\ell$ vertex-disjoint paths in $G$ such that the endpoints of the $j$-th path are $s_j$ and $t_j$, for each $j \in [\ell]$. A graph is called \emph{Hamiltonian-$\ell$-linked} if, in addition, the set of $\ell$ paths covers all vertices of $G$.

We need a procedure constructing the linkage for computing a TM-embedding of $H$.  For that we need a constructive version of the following theorem  of Thomas and Wollan~\cite{Thomas2005} whose original proof is non-constructive.

\begin{proposition}[Corollary~1.3, \cite{Thomas2005}]\label{prop:k-linked}
	If $G$ is $10\ell$-connected then $G$ is $\ell$-linked.
\end{proposition}

We do not know how to make the proof of \Cref{prop:k-linked} constructive in general. However, since the graphs of our interest have bounded independence number, we can use this property to construct linkages in highly-connected graphs.
The algorithm in the following lemma serves as a key starting subroutine in the proof of our result for highly-connected instances.

\begin{lemma}\label{lemma:disjoint_paths_or_is}
	There is an algorithm that, given integers $k,\ell\geq 1$, a $10\ell$-connected graph $G$, and a family of vertex pairs $(s_1, t_1), \ldots, (s_\ell, t_\ell)$ such that $s_i\neq t_i$ for each $i\in [\ell]$, in time $2^{(k+\ell)^{\Oh(k)}}+ |G|^{\Oh(1)}$ outputs either an independent set of size $k$ or a family  of $\ell$ internally  vertex-disjoint $s_i$-$t_i$-paths $P_i$ for $i\in[\ell]$, that is, no two paths have common internal vertices, and
	the $i$-th path
    does not contain any other $s_j$ or $t_j$ as an internal vertex.
\end{lemma}
\begin{proof}
Let $T=\bigcup_{i=1}^\ell\{s_i,t_i\}$ be the set of vertices in the input pairs.
We call the vertices in $T$ \emph{terminals}.	
For simplicity, our algorithm first ensures that $s_1,t_1,s_2,t_2,\ldots,s_\ell,t_\ell$ are $2\ell$ distinct vertices, i.e., $|T|=2\ell$.
For each vertex $v \in T$ that appears in $p>1$ input pairs, the algorithm replaces $v$ by $p$ true twins, i.e., by vertices with the same neighborhood as $v$ that are also adjacent to each other.
Then, each repeating occurrence of $v$ in the input pairs is replaced by a distinct copy.
This transformation results in $2\ell$ distinct terminal vertices.
The requirement on $10\ell$-connectivity of the graph still holds, and the modification does not increase the maximum size of an independent set.
Furthermore, given an independent set in the transformed graph, replacing each copy of a terminal vertex with the corresponding original vertex of $G$ gives an independent set in the original graph.
The same is true for paths avoiding terminal vertices: a path between copies is equivalent to the path between original vertices.
Therefore, the transformed input is equivalent to the original input, and the respective solutions can be transformed in polynomial time.
Henceforth, the algorithm works with the transformed input, meaning that all vertices in the input pairs are distinct.

Observe that by \Cref{prop:k-linked}, a family of vertex-disjoint $s_i$-$t_i$-paths always exists. Hence, our task is to construct such paths. 	
First, we show that if $G$ contains a clique of size $2\ell$, then the paths can be constructed in polynomial time.
	
	\begin{claim}\label{claim:disjoint_paths_by_clique}
		Let $C$ be a given clique of size $2\ell$ in $G$.
		Then a family of
		vertex-disjoint paths $P_1, P_2, \ldots, P_\ell$, where for each $i\in[\ell]$, $P_i$ connects $s_i$ with $t_i$, can be constructed in polynomial time.
	\end{claim}
	\begin{claimproof}
		Since $G$ is $2\ell$-connected, by Menger's theorem \cite{Menger1927,Diestel}, there are $2\ell$ vertex-disjoint paths connecting the vertices of $T=\bigcup_{i=1}^\ell\{s_i,t_i\}$ with distinct vertices of $C$.
		These paths can be constructed in the standard fashion via a network flow algorithm; in particular, using the Ford-Fulkerson algorithm \cite{Ford1956-ym} to find $2\ell$-flow in $\Oh(\ell \cdot |G|)$ running time.
		Then for each $i\in [\ell]$, we have two paths, one connecting $s_i$ with $c_{2i-1}$ and the other connecting $t_i$ with $c_{2i}$ for some $c_{2i-1}, c_{2i}\in C$.
		Note that it is possible that one or two of these paths are trivial, i.e., $s_i=c_{2i-1}$ or $t_i=c_{2i}$ can be true.
		
		As $C$ is a clique in $G$ and $c_{2i-1}c_{2i}\in E(G)$, these paths can be concatenated, resulting in the $s_i$-$t_i$-path $P_i$.
		Since $c_1, c_2, \ldots, c_{2\ell}$ are pairwise distinct, any two paths among $P_1, P_2, \ldots, P_\ell$ do not share any common vertices.
		We obtain $\ell$ paths as required.
	\end{claimproof}

When $G$ is of bounded size, we employ brute-force enumeration. 	
	
	\begin{claim}\label{claim:disjoint_paths_bruteforce}
		If $G$ has fewer than $(k+2\ell)^k$ vertices,  then  paths $P_1, P_2, \ldots, P_\ell$ can be found in $\Oh(2^{(k+2\ell)^{2k}})$  time.
	\end{claim}
	\begin{claimproof}
		It is sufficient to enumerate all possible subsets of edges of $G$ contained in the desired family of paths.
		When a subset $S$ of edges is fixed by the algorithm, it is sufficient to check that for each $i\in [\ell]$, there is a connected component in the graph with the vertex set $V(G)$ and the edge set $S$, that is a path between $s_i$ and $t_i$.
		If $S$ satisfies this for each $i\in[\ell]$, the algorithm stops and reports the corresponding components as $P_1, P_2,\ldots, P_\ell$.
	\end{claimproof}

	It is left to explain how the algorithm constructs the clique $C$, given that $G$ has at least $(k+2\ell)^k$ vertices.
	The essential tool here is the classical Ramsey theorem \cite{Ramsey_1930}, which guarantees a large clique or independent set in a large enough graph.
	For completeness, we present the corresponding folklore polynomial-time subroutine based on the Ramsey number bound given by Erd\H{o}s and Szekeres in 1935 \cite{Erds2009}.
	
	\begin{claim}\label{claim:ramsey_subroutine}
		There is a polynomial-time algorithm that, given a graph $G$ and two integers $r,s>0$ such that $|V(G)|\ge \binom{r+s-2}{r-1}$, outputs either an independent set of size $r$ or a clique of size $s$ in $G$.
	\end{claim}
	\begin{claimproof}			
		Note that $G$ necessarily has either an independent set of size $r$ or a clique of size $s$ because the corresponding Ramsey number satisfies $R(r,s)\le \binom{r+s-2}{r-1}$ \cite{Erds2009}.
		The algorithm is recursive.
		Since $r,s>0$, $|V(G)|\geq 1$ and $G$ has at least one vertex, denote it by $v$.
		If $r=1$ or $s=1$, then the algorithm outputs $\{v\}$ and stops.
		For $r,s\geq 2$, consider vertex sets $A=N_G(v)$ and $B=V(G)\setminus \{v\}\setminus A$.
		Since $|V(G)|=|A|+|B|+1$ and $|V(G)|\ge \binom{r+s-2}{r-1}$ we have that either $|A|\ge \binom{r+s-3}{r-1}$ or $|B|\ge \binom{r+s-3}{r-2}$.
		If $|A|\ge \binom{r+s-3}{r-1}$ then $G[A]$ has either an independent set of size $r>0$ or a clique of size $s-1>0$.
		In this case, the algorithm makes a recursive call to itself for the input $G[A]$, $r$ and $s-1$.
		If its output is a clique of size $s-1$, the algorithm additionally appends $v$ to it, resulting in a clique of size $s$.
		Thus, the algorithm obtains either an independent set of size $r$, or a clique of size $s$ as required.
		The case of $|B|\ge \binom{r+s-3}{r-2}$ is symmetric, as $v$ can be appended to any independent set in $G[B]$.
		
		Because the values of $\binom{x+y-2}{y-1}=\binom{x+y-2}{x-1}$ for $x\leq r$ and $y\leq s$ can be computed in $\Oh(rs)$ time using the recursive formula for binomial coefficients, 
		we obtain that the running time of the algorithm is $\mathcal{O}(n^2)$ where $n=|V(G)|$.
		The proof of the claim is complete.
	\end{claimproof}

	Now, we are ready to complete the algorithm:
	\begin{itemize}
		\item If $|V(G)|< (k+2\ell)^k$, the algorithm applies the subroutine of \Cref{claim:disjoint_paths_bruteforce} to find the solution, outputs it and stops.
		\item Otherwise, since $(k+2\ell)^k\ge\binom{k+2\ell-2}{k-1}$, the algorithm applies the subroutine of \Cref{claim:ramsey_subroutine} with $r=k$ and $s=2\ell$.
		\item If the output is an independent set of size $k$, the algorithm outputs it, and stops.
		\item Otherwise, the output is a clique of size $2\ell$ in $G$.
		The algorithm uses the subroutine of \Cref{claim:disjoint_paths_by_clique} to find the solution, then outputs it and stops.
	\end{itemize}
	
	The correctness of the algorithm and the upper bound on its running time follow from the claims.
	The proof of the lemma is complete.
\end{proof}

We now proceed to present an algorithm that, for a highly-connected graph $G$, a graph $H$, and a fixed mapping $f: V(H) \to V(G)$, either constructs a TM-embedding of $H$ in $G$ respecting $f$ and spanning $V(G)$, or finds an independent set of size $k$ in $G$.

\begin{lemma}[Spanning Lemma]
	Let $H$, $G$ be graphs, such that $H$ is non-empty.
	Let $f: V(H) \to V(G)$ be an injective mapping, and let $k\geq 1$ be an integer.
	Let $G$ additionally be either complete or $(\max\{k+2,10\}\cdot h)$-connected, where $h=|V(H)|+|E(H)|$.
	There is an algorithm with running time $2^{(h+k)^{\Oh(k)}}+|G|^{\Oh(1)}$ that computes either a subgraph $M$ of $G$ so that $(M, f)$ is a TM-embedding of $H$ in $G$ spanning $V(G)$, or an independent set of size $k$ in $G$.
	\label{lemma:highly-connected}
\end{lemma}
\begin{proof}
    The claim is straightforward if $G$ is a complete graph as $G$ always admits a spanning TM-embedding in this case. Indeed, since $f$ is injective, $f(x)$ and $f(y)$ are adjacent whenever $xy\in E(H)$,
    that is, $(M,f)$, where $M=(f(V(H)),\{f(x)f(y)\mid xy\in E(H)\})$, is a TM-embedding.
    Consider $U=V(G)\setminus f(V(H))$. If $U=\emptyset$, we already have a spanning TM-embedding. 
    Otherwise, $G$ has a path $P$ whose vertex set is $U$. Because $H$ is non-empty, $H$ has an edge $xy$. Then we replace $f(x)f(y)$ in the current embedding by the path obtained by making $f(x)$ and $f(y)$ adjacent to the endpoints of $P$. This creates a spanning TM-embedding.
    Thus, from now on, we assume that $G$ is not complete. 
    
	First, the algorithm reduces to the case where $H$ has no isolated vertices.
	Since each isolated vertex $v$ in $H$ has a fixed image $f(v)$ in $G$, the algorithm simply removes $v$ from $H$ and $f(v)$ from $G$.
	In the following steps of the algorithm, if a spanning TM-embedding of $H$ is produced, the algorithm extends this embedding with the removed images of isolated vertices.
	Such an embedding remains spanning for $G$, while all described procedures are performed in polynomial time.
	
	Note that removing an equal number of vertices from both $G$ and $H$ preserves the connectivity constraint in the statement.
	From now on, every vertex of $H$ has at least one incident edge, while $G$ is ($\max\{k+2,10\}\cdot |E(H)|$)-connected.
	
	In order to proceed, it is convenient to reformulate the problem in terms of disjoint paths:
	\emph{Given $\ell$ pairs $(s_1, t_1), (s_2, t_2), \ldots, (s_\ell, t_\ell)$ of distinct vertices in $G$, find a family of internally-disjoint paths $P_1, P_2,\ldots, P_\ell$ such that $P_i$ connects $s_i$ with $t_i$ for each $i\in[\ell]$ and each vertex of $G$ belongs to at least one path. }
	
	To see that the reformulation is equivalent, put $\ell=|E(H)|$ and $(s_i,t_i)=(f(x_i), f(y_i))$, where $x_i, y_i$ are the endpoints of the $i$-th edge of $H$.
	Then, given the paths $P_i$, it is straightforward to construct the subgraph $M$ by taking the union of the paths; $(M, f)$ then is the desired TM-embedding of $H$ in $G$.
	The reverse is also possible: given a TM-embedding $(M, f)$ of $H$ in $G$ spanning $V(G)$, the respective path can be found by traversing $M$ from $f(x_i)$ to $f(y_i)$.
	From now on, we do not refer to TM-embeddings, but work with the equivalent disjoint-path statement above.
	
	The first step of the algorithm is to find an initial family of paths $P_1, P_2, \ldots P_\ell$, that satisfies all constraints except for spanning $V(G)$.
	That is, $P_i$ is an $s_i$-$t_i$-path, and for each $i\neq j$,  $P_i$ and $P_j$ have no common vertices except for possibly their endpoints.
	For this, the algorithm invokes the algorithm of \Cref{lemma:disjoint_paths_or_is} as a subroutine.
	The call is valid since $G$ is $10\ell$-connected, therefore in time $2^{(k+\ell)^{\Oh(k)}}+|G|^{\Oh(1)}$ either an independent set of size $k$ or a sequence $P_1, P_2, \ldots, P_\ell$ is produced. In the firsts case, the algorithm outputs the obtained independent set and stops. Otherwise, we have an initial path family. 
	
	The only step left is to ensure that the spanning constraint is satisfied by the path sequence.
	The algorithm achieves this incrementally.
	That is, the algorithm enlarges one of the paths in the sequence without violating any other constraints.
	If we show how this can be done exhaustively and in polynomial time, it is clearly sufficient to find the desired TM-embedding.

We describe the enlargement process, akin to the method outlined in \cite{jedlickovaK23}, with a key distinction that unlike the scenario in \cite{jedlickovaK23}, we lack assurance that $G$ satisfies $\alpha(G)<k$, making the procedure potentially prone to failure. Therefore, we ensure that  (a) the enlargement is completed within polynomial time and (b) if the enlargement fails, an independent set of size $k$ is identified in $G$.

The procedure works as follows: If $P_1, P_2,\ldots, P_\ell$ span all vertices of $G$, the algorithm achieves the desired outcome, reports the solution, and halts. Otherwise, let $S=\bigcup_{i=1}^\ell V(P_i)$ denote the span of the path sequence, with the assumption that $S\neq V(G)$. The algorithm selects any $x\in V(G)\setminus S$. By the $(k+2)\ell$-connectivity of $G$, Menger's theorem guarantees the existence of $\min\{(k+2)\ell, |S|\}$ paths. Each path goes from  $x$ to $S$, and no two paths share any common vertex (including endpoints in $S$) but $x$. One can find such paths by utilizing the Ford-Fulkerson algorithm \cite{Ford1956-ym} to find a flow of size at most $(k+2)\ell$ in $\Oh(k\ell \cdot |G|)$ running time.

	If $|S|\le (k + 2) \ell$, then two of these paths end in $s_1$ and $v$, where $v$ is the neighbor of $s_1$ in $P_1$.
	Denote these paths by $Q_1$ and $Q_2$. These two paths are internally disjoint, both starting in $x$ and ending in $s_1$ and $v$, respectively.
	The algorithm enlarges $P_1$ by replacing it with the path $(s_1Q_1x)\circ(xQ_2v)\circ (vP_1t_1)$, that is, $x$ becomes embedded in $P_1$ between $s_1$ and $v$. Here and next, for a path $P$ and vertices $u$, $v$ on $P$ we denote by $uPv$ the subpath of $P$ from $u$ to $v$, and for two paths $P_1$ and $P_2$ sharing an endpoint, we denote by $P_1 \circ P_2$ their concatenation.
	This operation strictly increases the span of the path sequence.
	
	It is left for the algorithm to process the case $(k + 2) \ell < |S|$.
	There are precisely $(k+2)\ell$ paths that were found by the Ford-Fulkerson algorithm, so at least $k\ell$ paths end in internal vertices of $P_1, P_2, \ldots, P_\ell$.
	Hence, there is $i\in[\ell]$ such that at least $k$ paths end in internal vertices of $P_i$.
	The algorithm arbitrarily takes $k$ such paths.
	Denote the $k$ selected paths by $Q_1, Q_2, \ldots, Q_k$ and their respective endpoints in $P_i$ by $v_1, v_2, \ldots, v_k$.	For each $a \in [k]$, let $u_a$ be the successor of $v_a$ in $P_i$ if one follows vertices on $P_i$ from $s_i$ to $t_i$.
	
	If the set $\{u_1, u_2, \ldots, u_k\}$ is an independent set in $G$, then the algorithm reports this independent set and stops.
	Otherwise, there is an edge $u_a u_b$ in $G$ for some $a, b\in [k]$.
	Without loss of generality, $v_a$ goes before $v_b$ on $P_i$ in the order from $s_i$ to $t_i$.
	The algorithm then enlarges $P_i$ by replacing it with
	$(s_iP_iv_a) \circ (v_aQ_ax) \circ (xQ_bv_b) \circ (v_bP_iu_a) \circ u_au_b \circ (u_bP_it_i)$.
	Again, $x$ is embedded into $P_i$ while no original vertices of $P_i$ are lost, increasing the total span.
	
	The algorithm repeats the described procedure until paths $P_1, P_2, \ldots, P_\ell$ span the whole vertex set of $G$, or an independent set of size $k$ is encountered.
	This procedure is repeated at most $|V(G)|$ times, so the total running time of this part is polynomial in the size of the instance.
	
	The correctness of the algorithm follows from the discussion.
	The proof of the lemma is complete.
\end{proof}

Finally, we note that \Cref{lemma:highly-connected} admits an existential variant as a direct corollary: If $G$ is sufficiently well-connected compared to its independence number $\alpha(G)$ and the size of $H$, then a
spanning TM-embedding of $H$ exists in $G$ for any choice of $f$. As this may be of independent
interest, we state formally the following.

\begin{proposition}\label{prop:comb-highly-connected}
Let $H$ and $G$ be non-empty graphs, and let $f\colon V(H)\rightarrow V(G)$ be an injective mapping.
Let $G$ additionally be $(\max\{\alpha(G)+3,10\}\cdot h)$-connected, where $h=|V(H)|+|E(H)|$. There exists
a subgraph $M$ of $G$ such that $(M,f)$ is a TM-embedding of $H$ in $G$ spanning $V(G)$.
\end{proposition}

\subsection{Merging Lemma}

With the highly-connected case at hand, we now move to the second key ingredient in the proof of \Cref{thm:max_model}. We show that if the graph is not necessarily highly-connected, but contains a small-sized vertex subset such that after its removal the remaining components are highly-connected, then we can solve the instance in \classFPT time. Here, the parameter is the size of the subset to remove plus the size of $H$ plus $k$. In essence, the proof boils down to reconstructing the global solution from the solutions obtained on the highly-connected components via \Cref{lemma:highly-connected}, therefore we dub the result Merging Lemma, stated next.

\begin{lemma}[Merging Lemma]
	Let $G$ be a graph, let $S\subseteq V(G)$ be a vertex subset in $G$, and let $C_1$, \ldots, $C_t$ be the connected components of $G - S$.
    Let $H$ be a graph, $k\geq 1$ be an integer, and $L: V(H) \to 2^{V(G)}$ be a list assignment. Assume that each component $C_i$ is either a complete graph or at least  $5\max\{k + 2, 10\} \cdot (h + s)$-connected, where $h = |V(H)| + |E(H)|$ and $s = |S|$.
	There is an algorithm that either returns a maximum-size list TM-embedding of $H$ in $G$ respecting the list assignment $L$, or correctly reports that none exists, or returns an independent set in $G$ of size $k$. The running time of the algorithm is $2^{(s + h + k)^{\Oh(k)}} \cdot |G|^{\Oh(1)}$.
    \label{lemma:model_sep}
\end{lemma}
\begin{proof}

    \begin{figure}[h]
        \centering
        \includegraphics{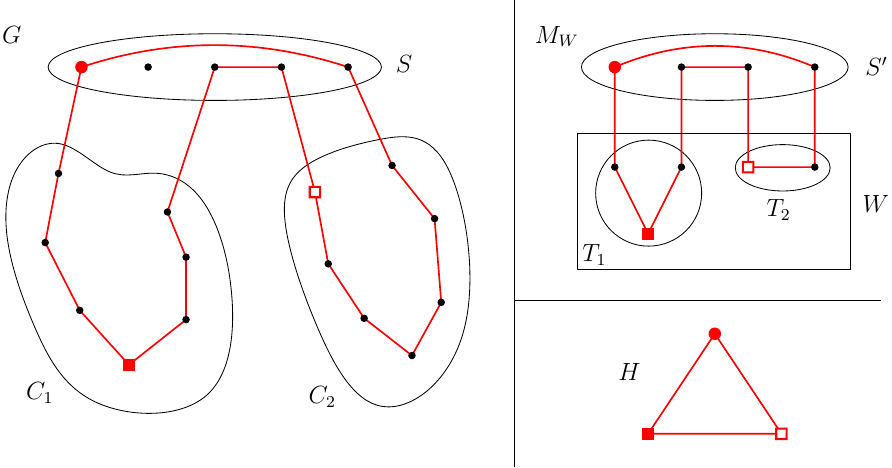}
        \caption{An illustration to the proof of Lemma~\ref{lemma:model_sep}: A graph $G$ with a TM-model of $H$, the graph $M_W$ of the cut descriptor corresponding to the model, and the graph $H$, the model of which is sought. The images of the vertices of $H$ in $G$ and $M_W$ that are fixed by $f$ and $f_W$, respectively, are indicated by the matching node shapes.}
        \label{fig:merging}
    \end{figure}

    \newcommand{\cutdesc}{cut descriptor\xspace}
    The main idea of the proof is to enumerate all possible interactions of the desired TM-model with $S$.
    We encapsulate the choice of the interaction in an object we call a \emph{\cutdesc}, to be formally defined next.
    Intuitively, a \cutdesc captures the behavior of a TM-model of $H$ in $G$ in the subset $S$ and its neighborhood: Given a TM-model $(M, T)$ of $H$ in $G$, its respective \cutdesc is obtained from $M$ by dissolving all non-terminals outside of $N[S]$, where vertices not in $S$ are represented implicitly by the index of their component. The latter part is crucial to keep the number of choices bounded by a function of the parameter. As we will see later, such an implicit representation is sufficient since \Cref{lemma:highly-connected} always provides a spanning TM-model (or an independent set) and allows to fix an arbitrary injection of the terminals. See   \Cref{fig:merging} for an illustration of the relationship between a TM-model of $H$ in $G$ and its \cutdesc.

    Formally, a \cutdesc is the tuple $(W, M_W, \xi, f_W)$, where $W$ is a set of vertices (indexed implicitly from $1$ to $|W|$), $M_W$ is a graph with the vertex set
    $S' \cup W$ with $S' \subseteq S$,
    $\xi: W \to [t]$ is a mapping from $W$ to indices of the components $C_1$, \ldots, $C_t$, and $f_W: V(H) \to S' \cup W$ is an injective mapping from $V(H)$ to vertices of $M_W$. Denote $T_W = \Ima f_W$, and for each $i \in [t]$, denote $T_i = \{w \in W : \xi(w)=i\}$. We require the following properties for $(W, M_W, \xi, f_W)$ to be a \cutdesc:
    \begin{enumerate}[label=\textbf{CD\arabic*}]
        \item\label{CD_model} $(M_W, f_W)$ is a TM-embedding of $H$ in $M_W$;
        \item\label{CD_W_concise} every $w \in W$ is in $T_W \cup N_{M_W}(S')$;
        \item\label{CD_W_components} for each $w, w' \in W$ with $\xi(w) \ne \xi(w')$, $w$ and $w'$ are not adjacent in $M_W$;
        \item\label{CD_S_in_G} $M_W[S']$ is a subgraph of $G$, and for each $x \in S', y \in V(H)$ with $f_W(y) = x$, $x \in L(y)$;
        \item\label{CD_W_exists} there exists an injective assignment $\tau: W \to V(G)$ such that i) $\tau(w) \in C_{\xi(w)}$, ii) if there exists $x \in V(H)$ with $f_W(x) = w$, then $\tau(w) \in L(x)$, iii) $N_{M_W}(w) \cap S' \subseteq N_G(\tau(w))$.
    \end{enumerate}

    We first observe an upper bound on the number of extra vertices in a \cutdesc.
    \begin{claim}
        If $(W, M_W, \xi, f_W)$ is a \cutdesc, then $|W|$ is at most $2(h + s)$.
    \end{claim}
    \begin{claimproof}
        We call a vertex $w \in W$ a \emph{terminal} if $w \in T_W$, and a \emph{connector} otherwise. From \ref{CD_W_concise}, every connector in $W$ is adjacent to a vertex of $S'$ in $M_W$. Since $f_W$ is injective, the number of terminals in $W$ is at most $|V(H)|$, therefore it remains to bound the number of connectors in $W$ by $2|E(H)| + 2s$. Let $w \in W$ be a connector that is adjacent to a terminal in $S'$, the number of such connectors is at most $\sum_{x \in S' \cap T_W} |N_{M_W}(x) \cap W| \le 2|E(H)|$, since by \ref{CD_model} every $x \in S' \cap T_W$ has degree $\deg_H(y)$, where $y$ is the unique pre-image of $x$ with respect to $f_W$.
        Finally, the number of connectors in $W$ adjacent only to non-terminals of $S'$ is at most $2s$, since every non-terminal has degree two in $M_W$ by \ref{CD_model} and $S' \subseteq S$.
    \end{claimproof}

    Next, we argue that if any list TM-embedding of $H$ in $G$ exists, then there exists also a ``matching'' \cutdesc. For a subgraph $M$ of $G$ and a subset $C \subseteq V(G)$, we say that $M$ \emph{hits} $C$ if $M[C]$ contains an edge, $M$ \emph{touches} $C$ if $M[C]$ contains no edges but at least one vertex, and $M$ \emph{avoids} $C$ if $V(M) \cap C = \emptyset$.

    \begin{claim}\label{claim:cutdesc_preserves}
        Let $(M, f)$ be a list TM-embedding of $H$ in $G$. Then there exists a \cutdesc $(W, M_W, \xi, f_W)$ such that $S' = V(M_W) \cap S = V(M) \cap S$, $M_W[S'] = M[S']$, and for each $i \in [t]$, $M$ hits $C_i$ if and only if $M_W$ hits $T_i$, touches $C_i$ if and only if $M_W$ touches $T_i$ and $|V(M) \cap C_i| = |T_i|$, and avoids $C_i$ if and only if $T_i = \emptyset$.
    \end{claim}
    \begin{claimproof}
        We construct the desired \cutdesc as follows. First, let $S' = V(M) \cap S$, and let $W$ be the subset of vertices in $G$ that are either terminals in $V(G) \setminus S'$ or neighbors of $S'$ in the model $M$, i.e., $W = (\Ima f \setminus S') \cup N_M(S')$. We then dissolve all vertices in $V(M) \setminus (S' \cup W)$, which are all non-terminals by definition of $W$, to obtain $M_W$ from $M$.
        We set $f_W$ as the respective mapping of the terminals obtained from $f$.
        Clearly, \ref{CD_model} is fulfilled for $(M_W, f_W)$, and \ref{CD_W_concise} also holds immediately by construction. We set $\xi$ so that $T_i = W \cap C_i$ for each $i \in [t]$. We now proceed to verify the remaining properties of a \cutdesc.

        For \ref{CD_W_components}, consider $w, w' \in W$ with $\xi(w) \ne \xi(w')$. Let $w \in T_i, w' \in T_j$ in $M_W$, $i \ne j$, and $w \in C_i$, $w' \in C_j$ in $G$. Since $S$ separates $C_i$ and $C_j$ in $G$, and $M$ is a subgraph of $G$, $w$ and $w'$ are not adjacent in $M$. By construction of $M_W$, only non-terminals outside of $S$ and $W$ are dissolved; every such non-terminal lies in $C_{i'}$ together with its neighbors, for some $i' \in [t]$, and dissolving it only changes edges between vertices of $C_{i'}$. Therefore, $w$ and $w'$ are also not adjacent in $M_W$.

        We now move to \ref{CD_S_in_G}, by construction $M_W[S'] = M[S']$ and is therefore a subgraph of $G[S]$.
Moreover, $f(y)=f_W(y)$ whenever $f(y)\in S'$; let $x$ be such that $f(y) = f_W(y) = x$, then $x \in L(y)$ since $(M, f)$ is a list TM-embedding respecting $L$.

Finally, for \ref{CD_W_exists}, the assignment $\tau$ is given directly by construction as an identity assignment on $W \subseteq V(G)$. For each $w \in W$, $\tau(w) \in C_{\xi(w)}$ by the construction of $\xi$. If $f_W(x) = w$, then $w = \tau(w) \in L(x)$ since $f(x) = w$ and $(M, f)$ is a list TM-embedding respecting $L$. By construction of $M_W$, the edges between $W$ and $S'$ are never changed between $M$ and $M_W$, therefore $N_{M_W}(w) \cap S' = N_M(w) \cap S'$, and the latter is a subset of $N_G(w)$ since $M$ is a subgraph of $G$.

    For the last part of the claim, we argue again by construction: since the dissolving operation only affects edges inside $C_i$, for some $i$, and dissolving a vertex in $C_i$ leaves $C_i$ with at least one edge inside, the status of the component (hit/touched/avoided) is always preserved between $M$ and $M_W$. Specifically, if an edge of $M$ was present in $C_i$, then after the dissolving, an edge is still present in $M_W$ in $T_i$. If no edge but some vertices of $M$ are present in $C_i$ then all these vertices are terminals or neighbors of $S'$, and $|M \cap C_i| = |T_i|$, 
    and if no vertices of $M$ are in $C_i$, then also $T_i$ is empty.
    \end{claimproof}

    Now, the algorithm proceeds as follows. First, it branches over the choice of a \cutdesc. Then, for a fixed \cutdesc, the algorithm picks an arbitrary assignment fulfilling \ref{CD_W_exists}.
    Since the assignment fixes the vertices of $W$ in $G$, we
    run the algorithm of \Cref{lemma:highly-connected} on each $C_i$ to find a TM-model connecting the vertices of $T_i$ as prescribed by $M_W$. If any of the runs returns an independent set of size $k$, we output it and stop. Otherwise, we augment the model $(M_W, T_W)$ to a model $(M, T)$ of $H$ in $G$ via the returned model in each $C_i$. Out of all models of $H$ in $G$ obtained in the branches, the algorithm returns the one of maximum size. The above algorithm is shown in detail in \Cref{alg:model_sep}.

    \begin{algorithm}
        \caption{Algorithm of \Cref{lemma:model_sep}.}
        \label{alg:model_sep}
        \KwData{Graphs $G$ and $H$, mapping $L: V(H) \to 2^{V(G)}$, integer $k$, vertex subset $S \subseteq V(G)$ splitting $G$ into components $C_1$, \ldots, $C_t$ such that for each $i \in [t]$, $G[C_i]$ is $5\max\{k + 2, 10\} (h + s)$-connected, where $h = |V(H)| + |E(H)|$, $s = |S|$.}
        \KwResult{An independent set of size $k$ in $G$, or a list TM-embedding of $H$ in $G$ respecting $L$, or the output that no such TM-embedding exists.}

        $\mathcal{M} \gets \emptyset$;

        \For{$|W| \gets 0$ \KwTo $2(h + s)$}{
            $W \gets $vertex set of size $|W|$;

            \ForEach{$S' \subseteq S$, $M_W$ a graph on $S' \cup W$, $\xi: W \to [t]$, $f_W: V(H) \to S' \cup W$}{\label{line:enumerate}
                verify \ref{CD_model}--\ref{CD_S_in_G}, verify \ref{CD_W_exists} by a matching instance and pick arbitrary suitable $\tau$;\label{line:matching}

                \If{any of \ref{CD_model}--\ref{CD_W_exists} fail}{
                    \Continue;
                }

                $M \gets M_W$, $f \gets f_W$;

                \For{$i \gets 1$ \KwTo $t$}{
                    $T_i \gets \{w \in W: \xi(w) = i\}$, $H_i = M_W[T_i]$;

                    \eIf{$H_i$ contains no edges}{\label{line:check_Hi}
                        $O \gets $ empty graph over the vertex set $\tau(T_i)$;\label{line:empty_Hi}
                    }{
                        $O \gets $ invoke \Cref{lemma:highly-connected} with $H_i$, $G[C_i]$, $\left.\tau\right|_{T_i}$, $k$;\label{line:highly_connected}
                    }

                    \eIf{$O$ is an independent set of size $k$}{
                        \Return $O$.
                    }{
                        $M_i \gets O$;
                        $M \gets M$ where $M[T_i]$ is replaced by $M_i$\label{line:replacement}\;
                        \ForEach{$w\in (T_i\cap\Ima f_W)$}{
                        	$f(f_W^{-1}(w))\gets \tau(w)$\;
                        }
                    }
                }

                $\mathcal{M} \gets \mathcal{M} \cup \{(M, f)\}$;\label{line:add_solution}
            }
        }

        \eIf{$\mathcal{M} \ne \emptyset$}{
            \Return $(M, f) \in \mathcal{M}$ with $\max |M|$.\label{line:return}
        }{
            \Return No solution.\label{line:no_solution}
        }
\end{algorithm}

It remains to verify that Algorithm~\ref{alg:model_sep} is correct and to upper-bound its running time.
The series of claims presented next verifies the correctness of the algorithm. While checking the conditions \ref{CD_model}--\ref{CD_W_exists} for a fixed \cutdesc in \Cref{line:matching} is generally straightforward, we provide more details for the condition \ref{CD_W_exists} by constructing $\tau$ directly.

\begin{claim}
  In \Cref{line:matching}, the condition \ref{CD_W_exists} can be verified, and a suitable assignment $\tau$ can be constructed via a call to a maximum matching algorithm on an auxiliary bipartite graph of polynomial size.
\end{claim}
\begin{claimproof}
    We construct an auxiliary bipartite graph $B$ where the two parts are $W$ and $V(G)$, and $w \in W$, $v \in V(G)$ are adjacent if setting $\tau(w) = v$ does not violate any of the conditions i--iii) of \ref{CD_W_exists}.
    That is, $w$ and $v$ are adjacent if and only if
    \begin{itemize}
        \item $v \notin S$ and $v \in C_{\xi(w)}$,
        \item if $w \in \Ima f_W$ with $f_W(x) = w$, then $v \in L(x)$,
        \item $N_{M_W}(w) \cap S' \subseteq N_G(v)$.
    \end{itemize}
    Clearly, a suitable assignment $\tau$ exists if and only if a matching in $B$ that covers $W$ exists.
\end{claimproof}

Then, we show that \Cref{lemma:highly-connected} may be invoked with the given arguments.

\begin{claim}
    \label{claim:call_highly}
 The call to \Cref{lemma:highly-connected} in \Cref{line:highly_connected}  is valid.
\end{claim}
\begin{claimproof}
By the condition in \Cref{line:check_Hi}, $H_i$ is not empty. The mapping $\left.\tau\right\vert_{T_i}$ is indeed an injective mapping from $T_i = V(H_i)$ to $C_i$, where the latter is guaranteed by \ref{CD_W_exists}, i). It remains to verify that $G[C_i]$ is $\max\{k + 2, 10\} \cdot (|V(H_i)| + |E(H_i)|)$-connected. By the statement of the lemma, $G[C_i]$ is $5\max\{k + 2, 10\} \cdot (h + s)$-connected, so it suffices to argue that $|V(H_i)| + |E(H_i)| \le 5(h + s)$.

    We consider two kinds of edges in $H_i$  separately. First, there are at most $2|E(H)|$ edges incident to vertices in $V(H_i) \cap T_W$, since for each $w \in T_W$, it holds that $\deg_{H_i}(w) \le \deg_{M_W}(w)$ and $\sum_{w \in T_W} \deg_{M_W}(w) \le 2|E(H)|$.
    Second, consider edges where both endpoints are connectors. By \ref{CD_model}--\ref{CD_W_concise}, each connector has degree two in $M_W$ and at least one neighbor in $S'$. Therefore, two connectors adjacent in $H_i$ have in total two distinct edges into $S'$. The size of the cut between $S'$ and $W$ in $M_W$ is at most $2s + 2|E(H)|$, since the total degree of terminals in $S'$ towards $W$ is at most $2|E(H)|$, and every other vertex has degree two. Thus, the number of edges of the second type in $H_i$ is at most $s + |E(H)|$. The total number of edges in $H_i$ is then at most $3|E(H)| + s$.

    For vertices, there are at most $|V(H)|$ vertices in $H_i$ that are in $T_W$, and every other vertex is a connector in $M_W$. As above, every connector has at least one neighbor in $S'$, and the cut size between $S'$ and $W$ is at most $2s + 2|E(H)|$. Therefore, the total number of vertices in $H_i$ is at most $|V(H)| + 2|E(H)| + 2s$, and $|V(H_i)| + |E(H_i)| \le |V(H)| + 5 |E(H)| + 3s \le 5(h + s)$.
\end{claimproof}

We are now able to show that we arrive at a TM-embedding of $H_i$ in $G[C_i]$, either after invoking \Cref{lemma:highly-connected} such that the output is not an independent set or after taking the empty graph over $\tau(T_i)$.

\begin{claim}
 In \Cref{line:replacement}, $M_i$ is such that $(M_i, \tau\vert_{T_i})$ is a TM-embedding of $H_i$ in $G[C_i]$.
\end{claim}
\begin{claimproof}
    In case $H_i$ is empty, in  \Cref{line:empty_Hi} the algorithm constructs the empty graph over $\tau(T_i)$, which is later assigned to $M_i$. Clearly, this is the desired TM-embedding.

    If $H_i$ is not empty, a call to  \Cref{lemma:highly-connected} is made in \Cref{line:highly_connected}. By  \Cref{claim:call_highly}, the call is valid, and therefore $O$ is either an independent set of size $k$ in $G[C_i]$, or the subgraph $M_i$ in the desired TM-embedding. Since the condition statement before \Cref{line:replacement} explicitly verifies that $O$ is not an independent set of size $k$, in \Cref{line:replacement} the latter case holds.
\end{claimproof}

In case when for each $i \in [t]$ we obtain the desired TM-embedding of $H_i$ in $G[C_i]$, we argue that the constructed $(M, f)$ is indeed a valid list TM-embedding of $H$ in $G$.

\begin{claim}
    At \Cref{line:add_solution}, $(M, f)$ is a list TM-embedding of $H$ in $G$ respecting $L$. 
\end{claim}
\begin{claimproof}
    The algorithm starts with the TM-model $M_W$ of $H$ in $M_W$, and replaces $M_W[T_i]$ by $M_i$ for each $i \in [t]$, where $M_i$ is the TM-model of $H_i$ in $G[C_i]$.
    The list condition for $M$ is fulfilled immediately since
    \[f(x)=\begin{cases} f_W(x),&\text{if }f_W(x)\in S',\\
    \tau(f_W(x)),&\text{if } f_W(x)\in W,
    \end{cases}\]
    $f_W$ respects $L$ on $S'$ via \ref{CD_S_in_G}, and $\tau$ respects $L$ via \ref{CD_W_exists}, ii). We now argue that $(M, f)$ is a TM-embedding of $H$ in $G$.

    First, we argue that $M$ is a subgraph of $G$ by considering three types of edges in $M$. Edges with two endpoints in $S'$ are not changed between $M_W$ and $M$, and therefore are present in $G$ since $M_W[S']$ is a subgraph of $G$ by \ref{CD_S_in_G}. Consider an edge between $x \in S'$ and $v \in C_i$, $i \in [t]$, in $M$. By construction of $M$, there exists an edge between $x$ and $w \in T_i$ in $M_W$ with $\tau(w) = v$. By \ref{CD_W_exists}, $N_{M_W}(w) \cap S' \subseteq N_G(\tau(w))$, meaning that $x$ and $v = \tau(w)$ are adjacent in $G$. Finally, consider an edge of $M$ with both endpoints outside of $S'$. By construction, the edge is a part of $M_i$, for some $i \in [t]$, and is present in $G$ since $M_i$ is a subgraph of $G[C_i]$.

    Second, we claim that all vertices in $V(M) \setminus \Ima f$ have degree two, and each $v \in \Ima f$ with $v = f(x_v)$ has $\deg_M(v) = \deg_H(x_v)$. Based on the above, the vertices of $S'$ do not change the degree between $M_W$ and $M$, and $M_W$ is a TM-model of $H$. Consider a vertex $v \in V(M) \setminus \Ima f \setminus S$, $v \in C_i$ for some $i \in [t]$. If $v \notin \tau(T_i)$, then $v$ has degree two in $M_i$ since $M_i$ is a TM-model of $H_i$, and has the same neighborhood in $M$ by construction of $M$. If $v \in \tau(T_i)$,  then its degree coincides with the degree of $w$ in $M_W$ with $\tau(w) = v$; the degree is also two since $M_W$ is a TM-model of $H$ and $w \notin \Ima f_W$. Finally, for a vertex $v \in \Ima f \setminus S$, $v \in C_i$ for some $i \in [t]$, $v = \tau(w)$ for some $w \in T_i \cap \Ima f_W$ and $\deg_{M_i}(v) = \deg_{H_i}(w)$ since $M_i$ is a TM-model of $H_i$. Similarly, $N_{M_W}(w) \cap S' = N_M(v) \cap S'$, therefore $\deg_M(v) = \deg_{M_i}(v) + |N_M(v) \cap S'| = \deg_{H_i}(w) + |N_{M_W}(w) \cap S'| = \deg_H(x)$, where $v = f(x), w = f_W(x)$, since $(M_W, f_W)$ is a TM-model of $H$ and all neighbors of $w$ in $M_W$ are either in $S'$ or in $T_i$ by \ref{CD_W_components}.
    
    We now argue that dissolving non-terminals in $M$ results in a graph isomorphic to $H$ via $f$. We call the graph obtained after dissolving by $H'$ and identify $V(H')$ with a subset of $V(M)$.
For  $x_v, x_u \in V(H)$ adjacent in $H$, we show that $f(x_v)$, $f(x_u)$ are adjacent in $H'$. Denote by $v = f(x_v)$, $u = f(x_u)$. If $v \in S'$, let $v' = v$, otherwise $v' \in W$ is such that $\tau(v') = v$. We define $u'$ analogously to $u$. By the choice of $v', u'$, we have that $f_W(x_v) = v'$, $f_W(x_u) = u'$. Also, since $(M_W, f_W)$ is a TM-embedding of $H$, $v'$ and $u'$ are connected by a path $P'$ in $M_W$ all of whose internal vertices are non-terminals.

    Based on $P'$, we construct a $v$-$u$-path $P$ in $M$ all of whose internal vertices are non-terminals.
    First, replace all vertices of $W$ in $P'$ with their images given by $\tau$. Second, for each two consecutive vertices $w_1, w_2$ in $P'$, such that $w_1, w_2 \in W$, insert in $P$ between $\tau(w_1)$, $\tau(w_2)$ the $\tau(w_1)$-$\tau(w_2)$-path in $M_i$, whose internal vertices are not in $\tau(T_i)$. Such a path exists because $w_1, w_2 \in T_i$ for some $i \in [t]$ by \ref{CD_W_components}, and because $M_i$ is a TM-model of $M_W[T_i]$ in $G[C_i]$. Since for any $s \in S'$ and  $w \in W$ that are neighbors in $M_W$, vertices $s$ and $\tau(w)$ are also neighbors in $G$ by \ref{CD_W_exists}, iii), $P$ is indeed a $v$-$u$-path in $G$. By construction, all internal vertices are not in $\Ima f$. Hence $v$ and $u$ are adjacent in $H'$. It remains to observe that $\deg_{H'}(v) = \deg_M(v) = \deg_H(x_v)$ for any $v, x_v$ with $v = f(x_v)$. Therefore, $H'$ is isomorphic to $H$ via $f$, completing the proof.
\end{claimproof}

Next, we show that trying each possible \cutdesc guarantees that the largest list TM-embedding is found.

\begin{claim}
    The TM-embedding $(M, f)$ returned on \Cref{line:return} is of maximum size among all suitable list TM-embeddings. 
\end{claim}
\begin{claimproof}
    Consider the list TM-embedding $(M^*, f^*)$ of $H$ in $G$ of maximum size. By \Cref{claim:cutdesc_preserves}, there exists a \cutdesc $(W, M_W, \xi, f_W)$ with $S' = V(M_W) \cap S = V(M^*) \cap S$, $M_W[S'] = M^*[S']$, and $M^*$ hits (touches/avoids) $C_i$ if and only if $M_W$ hits (touches/avoids) $T_i$ for each $i \in [t]$. Consider the TM-embedding $(M, f)$ constructed by the algorithm at \Cref{line:add_solution} with the \cutdesc set to $(W, M_W, \xi, f_W)$. Note that since the algorithm reached \Cref{line:return}, all calls to \Cref{lemma:highly-connected} resulted in a TM-model, and $M$ is successfully constructed. By construction, $V(M) \cap S = S'$, and $M[S'] = M_W[S'] = M^*[S']$. Observe also that if $M_W$ hits $T_i$, then the TM-model $M_i$ returned by \Cref{lemma:highly-connected} spans $C_i$, meaning that $|V(M) \cap C_i| = |C_i| \ge |V(M^*) \cap C_i|$. If $M^*$ touches $C_i$, then $|V(M^*) \cap C_i| = |T_i|$, and $M$ has at least as many vertices in $C_i$. If $M^*$ avoids $C_i$, $|V(M^*) \cap C_i| = 0$. Therefore, in each $S$, $C_1$, $C_2$, \ldots, $C_t$, $M$ has at least as many vertices as $M^*$, meaning $|V(M)| \ge |V(M^*)|$, proving the claim.
\end{claimproof}

We also argue that if the algorithm reports that no suitable TM-model is found, then indeed none exists.

\begin{claim}
    If $\mathcal{M}$ is empty by the end of the algorithm, then no list TM-model of $H$ in $G$ respecting $L$ exists. 
\end{claim}
\begin{claimproof}
    Assume a list TM-model of $H$ in $G$ exists. Then by \Cref{claim:cutdesc_preserves}, there exists also a \cutdesc. Consider the iteration of the algorithm with this \cutdesc. Since the algorithm reaches \Cref{line:no_solution}, all calls to \Cref{lemma:highly-connected} resulted in a suitable TM-embedding being found, therefore the algorithm necessarily adds a TM-embedding $(M, f)$ to $\mathcal{M}$ at \Cref{line:add_solution}. This contradicts the assumption that \Cref{line:no_solution} is reached.
\end{claimproof}

Finally, we estimate the running time. First note that all lines except \Cref{line:enumerate} and \Cref{line:highly_connected} incur at most a polynomial factor. By \Cref{lemma:highly-connected}, up to polynomial factors \Cref{line:highly_connected} takes time $2^{(h + s + k)^{\Oh(k)}}$, for each invocation. This holds since $|V(H_i)| + |E(H_i)| \le 5(h + s)$, see the proof of \Cref{claim:call_highly}. Moving to \Cref{line:enumerate}, let $w = 2h + 2s$ be the maximum size of $W$; there are
at most $2^s$ choices for $S' \subseteq S$, $t^w$ choices for $\xi$, and at most ${(s + w)}^{|V(H)|}$ choices for $f_W$. We may assume $t \le k$, since otherwise, an independent set of size $k$ can be found trivially by taking an arbitrary vertex from each $C_i$. We now argue that with a fixed $S'$ and $f_W$, there are at most $(s + w)^{2(s + w)}$ choices for $M_W$ that satisfy \ref{CD_model}. Indeed, terminals are fixed by $f_W$, therefore, it only remains to set the edges where at least one endpoint is a non-terminal. Since every non-terminal has degree two in $M_W$, there are at most $(s + w)^2$ choices for each non-terminal, and at most $(s + w)^{2(s + w)}$ choices in total. In total, \Cref{line:enumerate} performs at most $(s + h + k)^{\Oh(s + h)} = 2^{(s + h + k)^{1 + o(1)}}$ iterations. The overall running time is therefore at most $2^{(s + h + k)^{\Oh(k)}} \cdot |G|^{\Oh(1)}$.
\end{proof}

\subsection{Proof of \Cref{thm:max_model}}

We are now ready to show the main result of this section, which we restate for convenience.

\thmmaxmodel*
\begin{proof}
    \begin{algorithm}
        \caption{Algorithm of Theorem~\ref{thm:max_model}.}
        \label{alg:model_sep_thm}
        \KwData{Graphs $G$ and $H$, mapping $L: V(H) \to 2^{V(G)}$, integer $k$.}
        \KwResult{An independent set of size $k$ in $G$, or a list TM-embedding of $H$ in $G$ respecting $L$, or the output that no such TM-embedding exists.}

        $S \gets \emptyset$;

        \For{$j \gets 1$ to $\infty$}{
            $C_1$, \ldots, $C_t \gets $ connected components of $G \setminus S$;

            \If{$t \ge k$}{
                $I \gets \{v_1, \ldots, v_k\}$, where $v_i \in C_i$; 

                \Return $I$. \label{line:IS}
            }

            $\kappa_j \gets 5 \cdot \max\{k + 2, 10\} \cdot (|V(H)| + |E(H)| + |S|)$;

            \eIf{for each $i \in [t]$, $C_i$ is either $\kappa_j$-connected or complete}{
                \Break;\label{line:break}
            }{
               pick $i$ s.t. $C_i$ is neither $\kappa_j$-connected nor complete;

               find minimum separator $S'$ in $G[C_i]$;

               $S \gets S \cup S'$;\label{line:increase_S}
            }
        }

        $O \gets $ invoke \Cref{lemma:model_sep} on $G$, $H$, $L$, $k$, $S$;\label{line:invoke}

        \Return $O$.
\end{algorithm}

We iteratively construct the vertex set $S$ suitable for \Cref{lemma:model_sep}, then invoke \Cref{lemma:model_sep}.
    We start by setting $S = \emptyset$, then iteratively find a separator in one of the connected components of $G - S$, and add it to $S$.
    This happens until the connected components of $G - S$ are all $\kappa$-connected, where $\kappa = 5 \cdot \max\{k + 2, 10\} \cdot (|V(H)| + |E(H)| + |S|)$. If, at some point, there are at least $k$ connected components, we return the independent set constructed by taking a vertex from each connected component. Otherwise, when the connectivity condition is fulfilled, we run \Cref{lemma:model_sep} on the given input with the computed $S$, and return the output.
See  \Cref{alg:model_sep_thm} for a detailed algorithm description.

We now argue that the algorithm is correct. If the algorithm returns at \Cref{line:IS}, then the output is an independent set in $G$ of size $k$. Otherwise, the algorithm eventually reaches~\Cref{line:invoke}, since at every iteration of the main loop, the size of $S$ increases by at least one. It is left to verify that the call to \Cref{lemma:model_sep} is valid. For that, we need that every connected component of $G-S$ is $\kappa$-connected, for $\kappa = 5 \cdot \max\{k + 2, 10\} \cdot (|V(H)| + |E(H)| + |S|)$. However, this is exactly the condition at \Cref{line:break}, and this is the only way to leave the loop without returning at 
\Cref{line:IS}. Therefore, the algorithm's output is correct by \Cref{lemma:model_sep}.

It remains to upper-bound the running time of the algorithm. For that, it is crucial to upper-bound the size of $S$ when the algorithm reaches \Cref{line:invoke}. First, we observe that the main loop performs at most $k$ iterations before it stops either at \Cref{line:IS} or at \Cref{line:break}. Specifically, at the start of the $j$-th iteration, the number of connected components in $G - S$ is at least $j$, which implies that at iteration $k$, the loop necessarily stops at \Cref{line:IS}. This holds since each iteration increases the number of connected components by at least one, since a separator of one of the former components is added into $S$, while the remaining components are unchanged.

We now also observe that on the $j$-th iteration, the size of $S$ increases by at most $\kappa_j$, since we add to $S$ the minimum separator $S'$ of some $G[C_i]$, and $G[C_i]$ is not $\kappa_j$-connected. Let $s_j$ be the size of $S$ at the start of the $j$-th iteration, and let $\delta = 5 \cdot \max\{k + 2, 10\}$. We show the following bound.

\begin{claim}
    It holds that $s_j \le (\delta + 2)^j \cdot h$, where $h = |V(H)| + |E(H)|$.
\end{claim}
\begin{claimproof}
    We show the claim by induction. The claim holds for $j = 1$ since $s_1 = 0$.
    Assume the claim holds for $s_j$, consider $s_{j + 1}$:
    \begin{align*}
        s_{j + 1} \le & s_j + \kappa_j = s_j + \delta \cdot (s_j + h)= (\delta + 1) \cdot s_j + \delta \cdot h\\
        \le& (\delta + 1) \cdot(\delta + 2)^j \cdot h + \delta \cdot h \le (\delta + 2)^{j + 1} \cdot h,
    \end{align*}
    where the last inequality holds since $\delta \cdot h \le (\delta + 2)^j \cdot h$ for $j \ge 1$.
\end{claimproof}

Since the number of iterations is at most $k$, the size of $S$ by \Cref{line:invoke} is then at most $s_k = k^{\Oh(k)} \cdot h$.
We now upper-bound the running time of Algorithm~\ref{alg:model_sep_thm} by the running time bound of \Cref{lemma:model_sep}, since clearly up to \Cref{line:invoke} the algorithm takes polynomial time. By \Cref{lemma:model_sep}, the total running time is thus at most $2^{(hk^{\Oh(k)})^{\Oh(k)}} \cdot |G|^{\Oh(1)} = 2^{|H|^{\Oh(k)} \cdot k^{\Oh(k^2)}} \cdot |G|^{\Oh(1)}$.
\end{proof}

\section{Open questions}\label{sec:concl}

We conclude with several open research questions. All our results concern undirected graphs. The complexity of all these problems on \emph{directed} graphs of small independence number remains open. For example, we do not know whether \textsc{Hamiltonian Cycle} is \classNP-complete on {directed} graphs with an independence number $k$ for some fixed $k \geq 2$.  But we cannot exclude that the problem may be in \classXP or even \classFPT parameterized by $k$. Parameterized complexity offers a powerful toolbox for obtaining lower bounds for structural parameterization \cite{cut-and-count, DBLP:journals/talg/FominGLSZ19, FominGLS10, LokshtanovMS18}. However, all these reductions use gadgets with large independence numbers. Obtaining algorithmic lower bounds for parameterization with small independence number appears more challenging.

A similar set of questions is about the analogue of  \Cref{thm:above_GM} for directed graphs. The theorem of Gallai and Milgram holds for directed graphs.  However,  whether it could be extended algorithmically or not is not known. The ``simplest'' question here is whether there is a polynomial-time algorithm for deciding whether there is a path cover of size $\alpha(G) - 1$?

\addcontentsline{toc}{section}{Acknowledgements}
\section*{Acknowledgments}

Fedor V.\ Fomin and Petr A.\ Golovach were supported by the Research Council of Norway via the BWCA project (grant no.~314528).

Fedor V.\ Fomin was supported by the European Research Council (ERC) via grant NewPC (reference 101199930).  \begin{minipage}{0.2\textwidth}\includegraphics[width=0.9\textwidth]{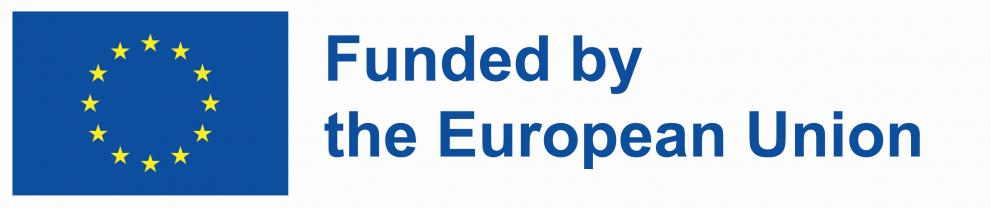}\end{minipage}

Petr A.\ Golovach was supported by the Research Council of Norway via the Extreme-Algorithms project (grant no~355137). 

Nikola Jedličková and Jan Kratochvíl were supported by GAUK 370122. 
    In the later stages of preparing this article, Nikola Jedličková was supported by the project PRIMUS/24/SCI/008 and Charles University Research Centre program No. UNCE/24/SCI/022.
    
    Danil Sagunov was supported by the Ministry of Science and Higher Education of the Russian Federation (agreement 075-15-2025-344 dated 29/04/2025 for Saint Petersburg Leonhard Euler International Mathematical Institute at PDMI RAS).

\section*{Declaration of generative AI and AI-assisted technologies in the manuscript preparation process}
The results were obtained and written without use of AI tools. ChatGPT and Claude were used in the final stage of the manuscript preparation to identify minor mistakes and inconsistencies and to correct grammatical errors. The authors take full responsibility for the correctness of the paper.

\section*{Declaration of competing interests}
The authors declare that they have no known competing financial interests or personal relationships that could have appeared to influence the work reported in this paper.


\end{document}

%% file: path-cover-transformations-fig.tex
\begin{figure}
	
	\begin{tabular}{c|c} 
		a 
		\begin{tikzpicture}[scale=0.7]
			
			\node[circle,draw,radius=0.25] (si) at (0,0) {};
			\node[above=0.05cm of si] {$s_i$};
			
			\node[circle,draw,radius=0.25] (ti) at (8,0) {};
			\node[above=0.05cm of ti] {$t_i$};
			
			\node[circle,draw,radius=0.25] (sj) at (0,-2) {};
			\node[below=0.05cm of sj] {$s_j$};
			
			\node[circle,draw,radius=0.25] (tj) at (8,-2) {}; 
			\node[below=0.05cm of tj] {$t_j$};
			
			\node[circle,draw,radius=0.25] (ui) at (4,0) {};
			\node[above=0.05cm of ui] {$u_i$};
			
			\node[circle,draw,radius=0.25] (uj) at (4,-2) {};
			\node[below=0.05cm of uj] {$u_j$};
			
			\node[circle,draw,radius=0.25] (vi) at (2.75,0) {};
			\node[above=0.05cm of vi] {$v_i$};
			
			\node[circle,draw,radius=0.25] (vj) at (5.25,-2) {};
			\node[below=0.05cm of vj] {$v_j$};
			
			\draw[very thick] (si) to[bend left] (ti);
			
			\draw[very thick, decorate,decoration={snake}] (si) -- node[above] {} (vi);  
			
			\draw (vi) -- node[above] {} (ui);
			
			\draw[very thick, decorate,decoration={snake}] (ui) -- node[above] {} (ti);
			
			\draw[very thick] (sj) to[bend right] (tj);
			
			\draw[very thick, decorate,decoration={snake}] (sj) -- node[below] {} (uj);
			
			\draw (uj) -- node[below] {} (vj); 
			
			\draw[very thick, decorate,decoration={snake}] (vj) -- node[below] {} (tj);
			
			\draw[very thick] (ui) -- (uj);
			
		\end{tikzpicture}
		
		&
	b 
		\begin{tikzpicture}[scale=0.7]
			
			\node[circle,draw,radius=0.25] (si) at (0,0) {};
			\node[above=0.05cm of si] {$s_i$};
			
			\node[circle,draw,radius=0.25] (ti) at (8,0) {};
			\node[above=0.05cm of ti] {$t_i$};
			
			\node[circle,draw,radius=0.25] (sj) at (0,-4) {};
			\node[below=0.05cm of sj] {$s_j$};
			
			\node[circle,draw,radius=0.25] (tj) at (8,-4) {};
			\node[below=0.05cm of tj] {$t_j$}; 
			
			\node[circle,draw,radius=0.25] (ui) at (4,0) {};
			\node[above=0.05cm of ui] {$u_i$};
			
			\node[circle,draw,radius=0.25] (uj) at (4,-4) {};
			\node[below=0.05cm of uj] {$u_j$};
			
			\node[circle,draw,radius=0.25] (vi) at (2.75,0) {}; 
			\node[above=0.05cm of vi] {$v_i$};
			
			\node[circle,draw,radius=0.25] (vj) at (5.25,-4) {};
			\node[below=0.05cm of vj] {$v_j$};
			
			\draw[red,very thick] (si) to[bend left] (ti);  
			
			\draw[red,very thick, decorate,decoration={snake}] (si) -- node[above] {} (vi);
			
			\draw (vi) -- node[above] {} (ui); 
			
			\draw[red,very thick, decorate,decoration={snake}] (ui) -- node[above] {} (ti);
			
			\draw[blue,very thick] (sj) to[bend right] (tj);
			
			\draw[blue,very thick, decorate,decoration={snake}] (sj) -- node[below] {} (uj);
			
			\draw (uj) -- node[below] {} (vj);
			
			\draw[blue,very thick, decorate,decoration={snake}] (vj) -- node[below] {} (tj);
			
			\node[circle,draw,radius=0.25] (sl) at (-2,-2) {};
			\node[below=0.05cm of sl] {$s_\ell$};
			
			\node[circle,draw,radius=0.25] (tl) at (10,-2) {};
			\node[below=0.05cm of tl] {$t_\ell$}; 
			
			\node[circle,draw,radius=0.25] (xi) at (3,-2) {};
			\node[below=0.05cm of xi] {$x_i$};
			
			\node[circle,draw,radius=0.25] (xj) at (7,-2) {};
			\node[below=0.05cm of xj] {$x_j$};
			
			\node[circle,draw,radius=0.25] (ni) at (4,-2) {};
			
			\node[circle,draw,radius=0.25] (nj) at (6,-2) {};
			
			\draw[red,very thick] (ui)--(xi); 
			
			\draw[blue,very thick] (uj)--(xj);
			
			\draw[red,very thick,decorate,decoration={snake}] (sl) -- node[above] {} (xi);
			
			\draw[very thick,decorate,decoration={snake}] (ni) -- node[above] {} (nj);
			
			\draw (xi) --(ni); 
			
			\draw (xj)--(nj);
			
			\draw[blue,very thick,decorate,decoration={snake}] (xj) -- node[above] {} (tl);  
			
			\draw[dashed] (sl) --(vi);
			
			\draw[dashed] (tl) --(vj); 
			
			\draw[dashed] (sl) to[bend left=12] (tl);
			
		\end{tikzpicture}
		
	\end{tabular}
	\caption{Illustration of transformations performed in (a) \Cref{rrule:merge_special_paths} and (b) \Cref{rrule:two_specials_to_one}.\label{fig:rrules}
	Paths before the transformations are horizontal.
	Paths after the transformations are highlighted in bold and have pairwise-distinct colors.
	Dashed lines represent absence of an edge.
	Snake lines represent paths instead of single edges.}
\end{figure}